\documentclass[11pt]{article}
\usepackage[dvips]{graphics}
\usepackage{epsfig}
\usepackage{multirow}
\usepackage{array}
\usepackage{extarrows}
\usepackage{graphicx}
\usepackage{subcaption}
\usepackage{comment}
\usepackage[titletoc,title]{appendix}
\numberwithin{equation}{section}
\usepackage{amsthm}
\usepackage{enumerate}
\usepackage{placeins}
\usepackage{amssymb,amsmath}

\usepackage{latexsym}
\usepackage[usenames]{color}
\usepackage{pgf}
\usepackage{appendix}
\usepackage{url}

\numberwithin{equation}{section}

\newcommand{\be}{\begin{equation}}
\newcommand{\ee}{\end{equation}}
\newcommand{\beaa}{\begin{eqnarray*}}
	\newcommand{\eeaa}{\end{eqnarray*}}
\newcommand{\bea}{\begin{eqnarray}}
\newcommand{\eea}{\end{eqnarray}}

\newcommand{\bei}{\begin{itemize}}
	\newcommand{\eei}{\end{itemize}}

\newtheorem{theorem}{ \noindent T{\footnotesize HEOREM}}
\newtheorem{prop}{ \noindent P{\footnotesize ROPOSITION}}[section]
\newtheorem{lemma}{ \noindent L{\footnotesize EMMA}}[section]
\newtheorem{coro}{ \noindent C{\footnotesize OROLLARY}}
\newtheorem{remark}{ \noindent R{\footnotesize EMARK}}[section]

\newtheorem{assumption}{Assumption}

\begin{document}
	\title{Fast Sampling and Inference via Preconditioned Langevin Dynamics}
	\date{}
	\author{Riddhiman Bhattacharya$^1$\and Tiefeng Jiang$^{2}$}
	
	\maketitle
\footnotetext[1]{Krannert School of Management, Purdue University, 403 W State St., West Lafayette, IN47907, USA, bhatta76@purdue.edu. 
.}
\footnotetext[2]{School of Statistics, University of Minnesota, Ford Hall, 224 Church St SE, Minneapolis, MN 55455, USA, jiang040@umn.edu
.}
 
 \begin{abstract}
    \noindent Sampling from distributions play a crucial role in aiding practitioners with statistical inference. However, in numerous situations, obtaining exact samples from complex distributions is infeasible. Consequently, researchers often turn to approximate sampling techniques to address this challenge. Fast approximate sampling from complicated distributions has gained much traction in the last few years with considerable progress in this field, for example, ~\cite{dalalyan2017theoretical,durmus2019high}. Previous work~\cite{roy2023convergence} has shown that for some problems a preconditioning can make the algorithm faster. In our research, we explore the Langevin Monte Carlo (LMC) algorithm and demonstrate its effectiveness in enabling inference from the obtained samples. Additionally, we establish a convergence rate for the LMC Markov chain in total variation. Lastly, we derive non-asymptotic bounds for approximate sampling from specific target distributions in the Wasserstein distance, particularly when the preconditioning is spatially invariant.
\end{abstract}
	\section{Introduction}
Sampling focuses on generating observations from a particular population for statistical analysis. Recently a lot of emphasis in literature has been on fast sampling for a large class of problems, for example, ~\cite{brosse2017sampling,dalalyan2020sampling,durmus2016sampling}. In this field of research, a notable trend is the observation of similarities between sampling and optimization methods. Researchers have successfully exploited this connection in numerous studies, resulting in the development of rapid sampling algorithms. These algorithms effectively generate observations from distributions characterized by densities in the form of \(\pi(x)\propto \exp\{-g(x)\}\) with \(\int \exp\{-g(x)\}\,dx <\infty\); see, for example,  \cite{dalalyan2017further,dalalyan2017theoretical,durmus2019high}.

In most literature the function $g(\cdot)$ is taken to be strongly convex with Lipschitz gradient. Recall that a function satisfies both  conditions if 
\begin{equation}\label{strngcnvxlipgrd}
\begin{split}
    & g(x)-g(y) \ge \left\langle \nabla g(y),x-y\right\rangle +\frac{m}{2}\left|x-y\right|^2~~ \text{and}\\
    & \left|\nabla g(x)-\nabla g(y)\right| \le M \left|x-y\right|
\end{split}
\end{equation}
for all $x, y \in \mathbb{R}^p$. The main intuition for this setting comes as follows; strong convexity and Lipschitz gradient are the conditions in which optimization algorithms work better. In fact, multiple algorithms in the fast sampling literature optimize the log-likelihood and use the mode as a warm start for the sampling algorithm~\cite{dalalyan2017theoretical}. The Langevin Monte Carlo algorithm has been studied in multiple works, for example,  \cite{dalalyan2017further,dalalyan2017theoretical,durmus2016sampling,durmus2019high,roberts2002langevin,roberts1996exponential} and is considered as the sampling analogue of the gradient descent algorithm.

There exist optimization algorithms that are variants of  gradient descent algorithm, achieved by applying a positive definite matrix as a preconditioner to the gradient. This technique ensures faster convergence, as demonstrated in previous studies~\cite{hendrikx2020statistically,li2017preconditioned}. Interestingly, these algorithms predate the gradient descent, with Newton's method being one of the earliest examples following this approach. We consider a similar setup for sampling by preconditioning the LMC algorithm with a fixed positive definite matrix. Similar work in literature can be found in the use of Metropolis Adjusted Langevin Algorithm using preconditioning and in some reinforcement learning setups, for example, ~\cite{10.2307/41057430,MAL-070,WANG2021110134}, which serve as the main motivation for our work. In fact, the work on preconditioned MALA~\cite{roy2023convergence} exhibit that preconditioning appropriately ensure a decrease in the effective sample size, which is defined as the number of samples from independent data having the same estimating power as that of a given number of correlated samples. A decrease in effective sample size in this case implies that implies that preconditioning in a ``correct fashion" ensures the requirement of lesser samples for inferential purposes. In reinforcement learning, existing literature has stated that the Langevin Monte Carlo when used for the purposes of Thomson Sampling exhibit faster mixing in practice under preconditioning~\cite{MAL-070}. Taking these works to serve as our motivation we consider the problem of establishing approximate sampling and inferential guarantees for the preconditioned Langevin Monte Carlo algorithm which has not been addressed in literature previously to the best of our knowledge.

Throughout this work we shall study the Langevin Monte Carlo  algorithm where the gradient and the noise are preconditioned using a function $H:\mathbb{R}^p \to \mathbb{R}^p \times \mathbb{R}^p$ which is positive definite at each point in $\mathbb{R}^p$. The equation for the algorithm is given by 
\begin{align}\label{precalgorithm}
    x_{k+1}=x_k-\gamma H(x_k)\nabla g(x_k)+\sqrt{2\gamma} H(x_k)^{1/2}\xi_{k+1}; \quad k=0,1,2,\cdots,
\end{align}
where $\gamma$ and $\xi_{k+1}$ are the step-size and Gaussian noise, respectively. 
Note that the matrix $H(x)^{1/2}$ is the usual square root matrix which is indeed well defined as $H(x)$ is positive definite. Also note that when $H(x)=I$, the identity matrix, the algorithm reduces to the standard LMC algorithm. In this work we consider two problems through the preconditioned LMC algorithm: the inference and the approximate sampling from distributions.  In both cases, we assume that $g(\cdot)$ satisfies \eqref{strngcnvxlipgrd}. In the former case we establish rates of convergence of the Markov Chain  \eqref{precalgorithm} to a stationary distribution dependent on $\gamma$ with respect to the total variation distance. We also establish a Central Limit Theorem (CLT) for the samples generated by \eqref{precalgorithm}. Note that in this regime $k \to \infty$, i.e., we may increase the number of iterations indefinitely and the approximation indeed becomes better as $k$ grows larger. In the second case, we establish that we may use \eqref{precalgorithm} with $H(\cdot)=H$ to sample from distributions with densities proportional to $\exp\{-g(x)\}$.  In this regime the number of iterations is bounded by some $K$, which is the maximum number of iterations of the algorithm permissible given a fixed step size $\gamma$. We show that when the maximum number of iterations is large and $\gamma$ is small such that the product $K\gamma=T$, the fixed time horizon, is large, the distribution of $x_k$ is close to the distribution corresponding the density $\pi(\cdot)$ (which we denote by $\Pi$) with respect to the Wasserstein metric. 

Now we make a summary about the new features of this paper. To the best of our knowledge, these two problems have not been addressed in previous literature and hence our analysis is probably the first try.
In our work, as mentioned previously, we consider the problem of inferential and approximate sampling guarantees using the preconditioned LMC algorithm. In this regard we establish a Central Limit Theorem for preconditioned LMC around the mode which may be used for the purposes of statistical inference. We also, in addition to this, establish \textit{explicit} convergence bounds of the algorithm to some stationary distribution in the Total Variation norm. We also establish approximate sampling bounds, in the Wasserstein distance, given a specific target as a function of the step size and the dimension. These results seem to be new in literature and are the main theoretical contributions of our paper.

We need two probability distances to state our main results in Section \ref{mainresults}. Recall that the total variation distance between two probability measures $\mu$ and $\nu$ is defined by
\[\|\mu-\nu\|_{TV}=\underset{A \in \mathcal{B}(\mathbb{R}^p)}{\sup} \left|\mu(A)-\nu(A)\right|\] 
where $\mathcal{B}(\mathbb{R}^p)$ is the set of all Borel sets in $\mathbb{R}^p.$ Also the
Wasserstein-Monge-Kantorvich distance
between $\mu$ and $\nu$ on $(\mathbb{R}^p, \mathcal{B}(\mathbb{R}^p))$ is defined by 
\begin{eqnarray}\label{Tie}
W_2(\mu,\nu)=\left(\inf_{\Tilde{\gamma}\in \Gamma(\mu,\nu)}\int_{\mathbb{R}^p\times\mathbb{R}^p} \left|x-y\right|^2 d\tilde{\gamma}(x,y)\right)^{1/2}
\end{eqnarray}
where the infimum is taken over all probability measure $\Tilde{\gamma} \in \Gamma(\mu,\nu)$ and $\Gamma(\mu,\nu)$ is the set of probability measures on $\mathbb{R}^p\times\mathbb{R}^p$ with  $\mu$ and $\nu$ being its marginal probability  measures on $\mathbb{R}^p$,  respectively. 

The rest of the paper is organized as follows. In Section~\ref{mainresults}, we state our  main results. In Section~\ref{exmpsim} we furnish some examples and provide some simulations corresponding to the examples. The technical proofs  are presented in  Appendices A and B. 

	\section{Main Results}\label{mainresults}
 As mentioned in the introduction, it has been seen in practise that preconditioning the LMC algorithm indeed speeds it up for some problems as in its optimization counterparts. Taking this to be our inspiration, we study the preconditioned Langevin Monte Carlo algorithm and  consider two problems in our work. The first is analyzing the preconditioned LMC algorithm in the regime where $k\to \infty$ in Section \ref{Shenzhen}. In this case our main objective is to ascertain whether the algorithm converges to some stationary distribution (which is not the target distribution) and whether we can use samples from the algorithm for inference on the mode of the target distribution. This will allow us to construct confidence intervals and carry out hypothesis tests. The second case is treated in Section \ref{Shanghai}, in which we sample up to some finite $K \in \mathbb{N}$, where $\mathbb{N}$ is the set of natural numbers, and then use them as an approximate sample from the target distribution. As for   preconditioning matrices appeared in both cases, in the first case we consider them to be spatially varying. In the second case, we choose them to be fixed matrices due to technical considerations.
 
 There has been some recent work on the analysis of preconditioned algorithms~\cite{roy2023convergence,gatmiry2022convergence,cheng2022theory}. These works mainly address the problem of establishing guarantees for fast sampling using preconditioned LMC in KL-divergence or in Wasserstein distance in the \textit{dissipative} setting and also establishing geometric ergodicity conditions for the purpose of sampling using preconditioned MALA. The novelty of our results in the fast sampling case is the existence of non-asymptotic bounds in the Wasserstein distance, in the strongly convex regime, in terms of the dimension and the step size which we believe are novel. In the case of inference the novelty of our results lie in establishing a Central Limit Theorem and also obtain exact convergence bounds for the convergence of the preconditioned LMC algorithm to a stationary distribution, in total variation, dependent on the step size. Again, we believe that these results have not been established for the preconditioned algorithm and hence provide some addition to the already rich literature of fast sampling.
 
 \subsection{ Inference from Preconditioned Langevin Monte Carlo}\label{Shenzhen}
Consider the algorithm 
\begin{align}\label{generalized:lmc}
x_{k+1}=x_k-\gamma H(x_k)\nabla g(x_k)+\sqrt{2\gamma}H^{1/2}(x_k)\xi_{k+1}
\end{align}
which is the preconditioned LMC algorithm with a spatially varying preconditioning matrix $H: \mathbb{R}^p \to \mathbb{R}^{p\times p}$. This algorithm has used predominantly in fast sampling when $H(x)=I$. We analyze this algorithm from an inferential perspective. We concentrate our efforts on three major points: a) does the algorithm have a stationary distribution and does it converge to that distribution? b) is the convergence geometric and can the rate be calculated? c) does a central limit theorem hold for the algorithm as observed in the case of SGD~\cite{robbins1951stochastic,10.1214/23-EJP947}?
Answering these questions is important as it allows us to know the reliability of our simulation and perform desired statistical tests.  We need the following standard assumptions for the rest of the paper.
	\begin{assumption}\label{assm1}
	The function $g(x)$ belongs to $C^2(\mathbb{R}^p)$ and is $m$-strongly convex with $M$-Lipschitz gradient i.e., it satisfies \eqref{strngcnvxlipgrd}.
	\end{assumption}
\begin{remark}\label{assm1:rmrk1}
    Note that {\it Assumption 1} implies \(mI\le \nabla^2 g(x) \le MI \), where $M>m>0$.
\end{remark}
 Here the notation $A\leq B$ for matrices $A$ and $B$ means that $B-A$ is non-negative definite. Remark~\ref{assm1:rmrk1} is well established and can be found in multiple previous works~\cite{dalalyan2017further}.
 
\begin{remark}
    Note that the strong convexity of $g$ ensures that it has a minimum~\cite{bottou2018optimization}. Let \(x^*=\arg\min g(x)\) be the minimum. 
\end{remark}

\begin{assumption}\label{assm:prec:bds}
	There exists $0<m_H\le M_H$ such that $m_H\, I \le H(x) \le M_H\, I$ for any $x\in \mathbb{R}^p$. Further, for any $x, y \in \mathbb{R}^{p}$, 
 \[\|H^{-1}(x)-H^{-1}(y)\|_2\le \frac{\beta}{M_H}\] for some constant $\beta$ with
 \[\beta<\frac{m_H^2\, m^2}{M_H^2\, M^2}\left(1-\frac{m_H^2\, m^2}{M_H^2\, M^2}\right)^{-1}.\]
\end{assumption}
The first part of Assumption~\ref{assm:prec:bds} is standard in literature. It requires the preconditioning matrix to be bounded both above and below. This guarantees that the algorithm does not blow up or remain static at any instance. The second part of Assumption~\ref{assm:prec:bds} implies that the change in the preconditioning matrix is upper bounded. The intuition for this is that  The exclusion of the two bounds means the condition number of the preconditioning matrix goes to infinity and hence the algorithm is hard to analyze. The following quantities will be needed to state our main results.
Define 
\begin{align}\label{key:quantities}
\begin{split}
    &\tilde{\lambda}=\left(1+\beta\right) \left(1-2\gamma \, m\,m_H -\gamma^2\, m^2_H\, m^2\right);\\    
    &\tilde{V}(x)=\left(x-x^*\right)^{\mathsf{T}} H^{-1}(x)\left(x-x^*\right)+1;\\ 
    & b=2\left(1+\beta\right)p\gamma; \\
    &\tilde{b}=b+1-\tilde{\lambda};\\
    & C=\left\{x:\, \tilde{V}(x)\le\frac{2\, \tilde{b}}{\alpha
		-\tilde{\lambda}} \right\};\\
     &\eta= \frac{\mu^{Leb}(C)}{(2\pi)^{p/2}M_H^{p/2}} \, \inf_{y \in C, x\in C}\exp\left\{-\frac{1}{2}\left(y-H(x)\nabla g(x)\right)^{\mathsf{T}}H^{-1}(x)\left(y-H(x)\nabla g(x)\right)\right\},    
\end{split}
\end{align}
where $\alpha$ is any value between $0$ and $\tilde{\lambda}$ and $\mu^{Leb}$ is the Lebesgue measure on $\mathbb{R}^p$. Also denote by $P$ the one-step Markov kernel for the Markov chain as defined by \eqref{precalgorithm}.  This seems to abuse the notation for the probability of an event, however, it will be evident from the context.
\begin{theorem}\label{thm:geom:erg}
	Let Assumptions~\ref{assm1} and~\ref{assm:prec:bds} hold and let 
	\[ \frac{m\, m_H}{M^2_H\, M^2}\left(1-\sqrt{1-\frac{M^2_H\, M^2\, \beta}{m^2_H\, m^2\, (1+\beta)}} \right)<\gamma< \frac{m\, m_H}{M^2_H\, M^2}\left(1+\sqrt{1-\frac{M^2_H\, M^2\, \beta}{m^2_H\, m^2\, (1+\beta)}} \right). \] Then the Markov chain as defined by \eqref{precalgorithm} has a stationary distribution dependent on $\gamma$, $\pi_{\gamma}$, and is $(M,\rho)$ geometrically ergodic with 
	\begin{align*}
	\|P^k (x,\cdot)-\pi_{\gamma}(\cdot)\|_{TV}\le M(x)\rho^k
	\end{align*}
	where \[M(x)=2+\frac{\tilde{b}}{1-\tilde{\lambda}}+\tilde{V}(x)\] and \[\rho\le \max\left\{(1-\eta)^r, \left(\frac{1+2\tilde{b}+\tilde{\lambda}+\tilde{\lambda}d}{1+d}\right)^{1-r}\left(1+2\tilde{b}+2\tilde{\lambda}d\right)^r\right\}\] for some free parameter $0<r<1$.
\end{theorem}
\begin{proof}
The proof Theorem \ref{thm:geom:erg} is furnished in the Appendix A.     
\end{proof}

\begin{remark}
    Theorem~\ref{thm:geom:erg} establishes a geometric convergence rate of the LMC Markov chain to some stationary distribution which is dependent on the step size of the algorithm.
\end{remark}

\begin{prop}
\label{coro:geom:erg}
    Let Assumptions~\ref{assm1} and~\ref{assm:prec:bds} hold and let $H(x)=H$ for all $x \in \mathbb{R}^p$, {namely, $H(x)$ is a constant matrix}. Then Theorem~\ref{thm:geom:erg} holds with $0<\gamma< \frac{2\,m\, m_H}{M^2_H\, M^2}$,
        $\tilde{\lambda}=1-2\gamma m\,m_H -\gamma^2 m\, m_H$ and $b=2\gamma p$ with $\tilde{b}$ as defined in \eqref{key:quantities}.
\end{prop}

\begin{proof}
The proof is provided in Appendix~\ref{sec:disc:time}.
\end{proof} 

\begin{remark}\label{eta:rho:remark:1}
  Note that $\eta$ and hence $\rho$ can be calculated/approximated if we can calculate or approximate $\mu^{Leb}(C)$ and $\underset{y,x \in C}{\inf}\,  \exp\left\{-\frac{1}{2}\left(y-H(x)\nabla g(x)\right)^{\mathsf{T}}H^{-1}(x)\left(y-H(x)\nabla g(x)\right)\right\}$ the latter of which is a hard problem. 
\end{remark}
\begin{prop}\label{prop:lwr:bnd}
    The statement of Theorem~\ref{thm:geom:erg} holds with 
    \begin{align*}
      \eta&=\frac{\mu^{Leb}(C)}{(2\pi)^{p/2}M_H^{p/2}} \, \exp\left\{-\left(\frac{2\tilde{b}}{\alpha -\tilde{\lambda}}-1\right)\left(\frac{\,M_H}{m_H}+\, M\,M_H +\frac{M^2_H\, M^2}{2} \right)\right.\\
      &\left. \quad\quad \quad \quad  -\frac{\|x^*\|^2}{m_H}- M\, \|x^*\| \sqrt{M_H\left(\frac{2\tilde{b}}{\alpha -\tilde{\lambda}}-1\right)}\right\}.
    \end{align*}
\end{prop}
\begin{proof}
    The proof is provided in the Appendix.
\end{proof}
\begin{remark}
    Proposition~\ref{prop:lwr:bnd} ensures a loose upper bound for the rate of convergence that may be \textit{explicitly} calculated. However; for a tighter upper bound on the rate of convergence one should indeed solve the optimization problem as mentioned in Remark~\ref{eta:rho:remark:1}
\end{remark}
Note that the value of $\mu^{Leb}(C)$ may be easily calculated by drawing randomly from the ball centred at $x^*$ with radius $\sqrt{M_H\,(2\tilde{b}/\alpha-\tilde{\lambda}-1)}$ and then accepting the number of samples that fall in $C$. Multiplying the proportion of accepted by the volume of the ball should provide an estimate of $\mu^{Leb}(C)$.

One also notes that all values of the free parameter result in a value less than $1$. Selecting an optimum value for $r$ is also a hard problem. One recommends practitioners to use multiple values of the free parameter in practise to find which works best.
Next we present a Central Limit Theorem for the samples from \eqref{generalized:lmc} which may be used for inferential purposes.
\begin{theorem}\label{thm:prec:clt}
		Let Assumptions~\ref{assm1} and~\ref{assm:prec:bds} hold and let 
			\[ \frac{m\, m_H}{M^2_H\, M^2}\left(1-\sqrt{1-\frac{M^2_H\, M^2\, \beta}{m^2_H\, m^2\, (1+\beta)}} \right)<\gamma< \frac{m\, m_H}{M^2_H\, M^2}\left(1+\sqrt{1-\frac{M^2_H\, M^2\, \beta}{m^2_H\, m^2\, (1+\beta)}} \right). \]
		Then for any function $f: \mathbb{R}^p\to \mathbb{R}$ with $f^2(x)\le V(x)$, we have
		\begin{align*}
		\sqrt{k}\left(\frac{1}{k}\sum_{i=0}^{k-1} f(x_i)-\int f\, d\pi_{\gamma}\right)\overset{d}{\rightarrow} N(0, \sigma^2(f,\pi_{\gamma}))
		\end{align*}
		as $k \to \infty$ where $\sigma^2(f,\pi_{\gamma})=\lim_{k\to \infty}\mathbb{E}_{\pi_{\gamma}}(\frac{1}{\sqrt{k}}\sum_{i=0}^{k-1} f(x_i))^2$ . 
\end{theorem}
\begin{proof}[Proof of Theorem~\ref{thm:prec:clt}]
	The proof of this theorem is immediate from Theorem~\ref{thm:geom:erg} and a previous result~\cite[Theorem 17.0.1]{meyn2012markov}. Note that it may be the case that $\sigma^2(f,\pi_{\gamma})=0$ and in this case, the result still holds if we consider $N(0,0)$ as the degenerate random variable with probability $1$ at $0$.  
\end{proof}

\begin{coro}
    		Let Assumptions~\ref{assm1} and~\ref{assm:prec:bds} hold and let 
			\[ \frac{m\, m_H}{M^2_H\, M^2}\left(1-\sqrt{1-\frac{M^2_H\, M^2\, \beta}{m^2_H\, m^2\, (1+\beta)}} \right)<\gamma< \frac{m\, m_H}{M^2_H\, M^2}\left(1+\sqrt{1-\frac{M^2_H\, M^2\, \beta}{m^2_H\, m^2\, (1+\beta)}} \right). \]
		Then for any $u \in \mathbb{R}^p$, with $\|u\|=1$, we have
		\begin{align*}
		\sqrt{k\,M^{-1}_H}\left(\frac{1}{k}\sum_{i=0}^{k-1} \left\langle u,x_i-x^*\right\rangle- \int\left\langle u,v-x^*\right\rangle d\pi_{\gamma}(v)\right)\overset{d}{\rightarrow} N(0, \sigma^2(f,\pi_{\gamma}))
		\end{align*}
		as $k \to \infty$ where $\sigma^2(f,\pi_{\gamma})=\lim_{k\to \infty}\mathbb{E}_{\pi_{\gamma}}(\frac{1}{\sqrt{k\, M_H}}\sum_{i=0}^{k-1} \left\langle u,x_i-x^*\right\rangle)^2$ 
\end{coro}

\begin{proof}
 The proof follows immediately by taking $f(x)=M^{-1/2}_H\langle u,x-x^*\rangle$ and noting that \[f^2(x) = M^{-1}_H \langle u,x-x^*\rangle^2 \le M^{-1}_H \|x-x^*\|^2\le (x-x^*)^{\mathsf{T}} H^{-1}(x)(x-x^*).\]
 Thus the result follows using Theorem~\ref{thm:prec:clt}.
\end{proof}
\begin{remark}
Note, this immediately implies that all one dimensional projections have a Central Limit Theorem.    
\end{remark}

\begin{prop}\label{prop:clt:constant:H}
     Let Assumptions~\ref{assm1} and~\ref{assm:prec:bds} hold and let $H(x)=H$ for all $x \in \mathbb{R}^p$. Then Theorem~\ref{thm:prec:clt} holds with $0<\gamma< \frac{2\,m\, m_H}{M^2_H\, M^2}$. 
\end{prop}
\begin{proof}
    The proof follows from Proposition~\ref{coro:geom:erg} and  Theorem 17.0.1 in \cite{meyn2012markov}.
\end{proof}

 \subsection{Approximate Sampling from a Specified Target}\label{Shanghai}
 In this section we focus on employing preconditioned LMC in an effort to sample from distributions with densities proportional to $\exp\{-g(x)\}$. For this point on we shall consider $H(x)=H$, i.e., the preconditioning is fixed. The reasoning for considering the algorithm as in \eqref{precalgorithm} with $H(x)=H$ is natural as it is the Euler discretization of the diffusion 
	\begin{align}\label{precdiffusion}
	dX_t=-H\nabla g(X_t)dt+ \sqrt{2}H^{1/2}dB_t
	\end{align}
	where $B_t$ is the standard Brownian motion.
 We define $\kappa=2m\,m_H$ and $\kappa_*=m_H/M_H$ the condition number of $H$. These quantities shall be used throughout the following this and the following chapters and are key quantities as expressed in our bounds later.
 
	Having completed stating our assumptions, the immediate question that arises is-does \eqref{precdiffusion} have the correct stationary distribution? The answer to the question is yes, it indeed does. Past work has indicated that diffusions of the form 
	\begin{align}\label{foxdiff}
	dz_t=f(z_t) dt +\sqrt{2\, D(z_t)}\, dB_t
	\end{align}
	have the correct stationary distribution subject to certain constraints for $f(\cdot)$.
	\begin{theorem}\label{thmfoxma}\cite[Theorem 1]{NIPS2015_9a440050}
		\(p^s(z)\propto \exp(-g(z))\) is a stationary distribution of \eqref{foxdiff} if 
		\begin{align*}
		f(z)&=-\left[D(z)+Q(z)\right]\nabla g(z)+\Gamma(z); \quad \text{where}\\
		\Gamma_i(z)&=\sum_{j=1}^{p}\partial_j \left(D_{ij}(z)+Q_{ij}(z)\right)
		\end{align*}
		with $D(z)$ positive semidefinite and $Q(z)$ skew-symmetric. If $D(z)$ is positive definite or if the ergodicity can be shown, then the stationary distribution is unique. 
	\end{theorem}
	Note that it immediately follows from Theorem~\ref{thmfoxma} that we indeed have the unique stationary distribution as $\pi(x)\propto \exp(-g(x))$. This is easy to see as for \eqref{precdiffusion}, we have $D(z)=H, \ Q(z)=0, \ \text{and }\Gamma(z)=0$ and hence the result follows. Note that considering a spatially varying preconditioning makes the problem much more complicated. There has been recent work~\cite{gatmiry2022convergence} where the authors show convergence of LMC in a Riemannian Manifold with respect to the KL-divergence; however, the assumptions used by the authors is much stronger than what we use.
	Given that we indeed have the correct stationary distribution for the process \eqref{precdiffusion}, the next natural question is whether the law of the process converges to the stationary distribution. In this regard, there has been previous work for similar problems \cite{dalalyan2017further,dalalyan2017theoretical,durmus2016sampling} and we follow in their footsteps.
	
		The following proposition concludes that the continuous time Markov Chain associated with \eqref{precdiffusion} is indeed geometrically converging in the Wasserstein distance where $\Pi$ denotes the stationary distribution.
	\begin{prop}\label{conttimeprop}
		Under Assumptions~\ref{assm1} and~\ref{assm:prec:bds} we have
		\begin{align*}
		W_2\left(\delta_x P_t,\Pi\right)\le  \frac{1}{\sqrt{\kappa_*}}e^{-\kappa t}\left(\left|x-x^*\right| +\left(\frac{p}{m\, \kappa_*}\right)^{1/2}\right)
		\end{align*} 
		where $\kappa$ and $\kappa_*$ are defined immediately following Assumption~\ref{assm:prec:bds}.
	\end{prop}  
 \begin{proof}
     The proof of Proposition~\ref{conttimeprop} is furnished in the Appendix.
 \end{proof}

	Note that this implies we converge exponentially to the stationary distribution with the rate being effected by how far away we are from the mode and there is a \(\sqrt p\) dependence on dimension. Also note that strong convexity is vital for our proof by noting the dependence of $m$ on the bound. Also note that the rate depends on the condition number of $H$, $\kappa_*$. 

	Next we establish convergence bounds for the Euler discretization 
	\begin{align*}
	x_{k+1}=x_k-\gamma H\nabla g(x_k)+\sqrt{2\gamma}H^{1/2}\xi_{k+1}
	\end{align*}
	of \eqref{precdiffusion} to the stationary distribution of \eqref{precdiffusion}. Define the Ito process in the time interval $[0,T]$ as 
	\begin{align}\label{discrprocess}
	dD_t= \sum_{k=0}^{[T/\gamma]}-H\nabla g(D_{k}) I(t)_{[k\gamma,(k+1)\gamma)} dt+\sqrt{2H}\, dB_t.
	\end{align}
Note that the marginals of \eqref{discrprocess} at time points $k\gamma$ have the same law as \eqref{precalgorithm}, that is, \(D_{k\gamma}\overset{d}{=} x_k\), where $\overset{d}{=}$ implies having the same law. Use $\mathbb{P}_{D_t}$ to denote the measure of the random variable $D_t$. \\
Define $L_t=\exp(\alpha\, X^{\mathsf{T}}_tH^{-1}X_t )$
and 
\begin{align}\label{Cstardef}
C^*=\frac{M_H}{6}\, \left(g(x_0)-g(x^*)\right)+\frac{5}{12}\,M^2_H\, MpT.
\end{align}
\begin{theorem}\label{Mainthm}
Under Assumptions~\ref{assm1} and~\ref{assm:prec:bds},  $0<\alpha<\kappa/2$, $\epsilon>0$ with
\[T=\frac{1}{\kappa} \log \left[\frac{2}{\epsilon}\left(\frac{1}{\sqrt{\kappa_*}}\left|x_0-x^*\right|+\frac{1}{\kappa_*}\left(\frac{p}{m}\right)^{1/2}\right)\right],\]
$C^*$ as defined in \eqref{Cstardef}, \[C=\frac{M_H}{\alpha}\left[\frac{3}{2}+2\alpha \, T\,p+\log\left(E\left(L_0\right)+\frac{\alpha \left|\nabla g(0)\right|^2}{\left(2 m- \frac{4\alpha}{m_H}\right)}\, T\right)\right]\]
and \[\gamma <\min \left(\frac{1}{32\, C^*}\left(\sqrt{1+\frac{2^{3/2}\epsilon}{C}}-1\right)^4, \ \frac{\kappa_*}{M M_H}\right),\] we have
\[W_2(\mathbb{P}_{D_t}, \Pi) \le \epsilon.\]
\end{theorem}
\begin{proof}
    The proof of Theorem~\ref{Mainthm} is furnished in the appendix.
\end{proof}

\begin{remark}
    Note that the smaller we consider $\epsilon$, either $T$ needs to be increased or $\gamma$ needs to be decreased.
    This implies that for large time horizon $T$ and small step size we shall  sample from the target distribution $\Pi$ with a small error. 
\end{remark}
\begin{remark}
    Note that the selection of the step size $\gamma$ is a sensitive problem for if the step-size is selected too small the algorithm needs time to explore the space. However, with the step size too large the algorithm returns incorrect results. Hence there should be an optimal step size which should vary depending the nature of the problem and preconditioning matrix used. 
\end{remark}
\section{Examples and Simulations}\label{exmpsim}
We consider three different examples in the following sections to exhibit the properties of preconditioned LMC as espoused by us in the previous sections. In all three examples we shall consider a preconditioning as defined by \[\mathcal{T}=\begin{bmatrix}
1 & \rho & \rho^2 &\cdots & \rho^{p-1}\\
\rho & 1 & \rho & \cdots & \rho^{p-2}\\
\vdots & \vdots& \vdots & \cdots & \vdots\\
\rho^{p-1} & \rho^{p-2}& \rho^{p-3} &\cdots & 1
\end{bmatrix}.\]
This is the first order AR(1) matrix where the value of $\rho$ determines it's eigenvalues which are all positive. In our problems we consider different values of $\rho$ and exhibit how it affects our simulation findings. Our algorithm for each example obeys the following update rule 
\[x_{k+1}=x_k-\gamma \mathcal{T}\nabla g(x)+\sqrt{2\gamma } \mathcal{T}^{1/2}\xi_{k+1}\]
where $\xi_k \sim N(0,I)$ for all $k=1,2,\cdots$ whete $\mathcal{T}$ is as previously espoused.

\subsection{Simulating from a Mixture Gaussian}
We consider the task of sampling from the density 
\[\pi(x)=\frac{1}{2\left(2\pi\right)^{p/2}}\left(e^{-\left|x-a\right|^2 /2}+e^{-\left|x+a\right|^2 /2}\right), \ x\in \mathbb{R}^p\]
Note that in this case 
\begin{align*}
    g(x)&=\frac{1}{2}\left|x-a\right|^2-\log\left(1+e^{-2 x^{\mathsf{T}}a}\right),\\
    \nabla g(x) &= x-a + 2a \, \left(1+e^{2 x^{\mathsf{T}}a}\right)^{-1}\\
    \nabla^2 g(x)&=I - 4 a\,a^{\mathsf{T}}\, e^{2 x^{\mathsf{T}}a}\left(1+e^{2 x^{\mathsf{T}}a}\right)^{-2}
\end{align*}
We note that $g(\cdot)$ is Lipschitz with Lipschitz-constant $1$. Also, when $|a|<1$, we have $m=1-|a|^2$. 
We exhibit simulations for different values of $\rho$ for which we plot histograms exhibiting the Central Limit Theorem for the spatial average of the observations generated from sampling. We also plot histograms exhibiting the approximate sampling from mixture normality. We generate $10^4$ observations with $10^3$ replicates for this study. 

Note that the approximate sampling of the marginal needs more samples if $\rho$ is taken closer to $1$. This can be explained by observing the eigenvalues of $\mathcal{T}$. As  can be found in previous literature~\cite{trench1999asymptotic}, the spectra of matrices of the form $\mathcal{T}$ is of the form  \[\lambda_{kn}=\frac{1-\rho^2}{1-\rho \cos\theta_{kn}+\rho^2}\] where $(k-1)\pi/n<\theta<k\pi/n$. Therefore as $\rho$ is taken closer to $1$, the eigenvalues of the matrix $\mathcal{T}$ become smaller and hence more data is needed to achieve the same error. 
\FloatBarrier
\begin{figure}
\captionsetup[subfigure]{justification=centering}
\centering
\subfloat{
\includegraphics[width=15cm, height=4cm]{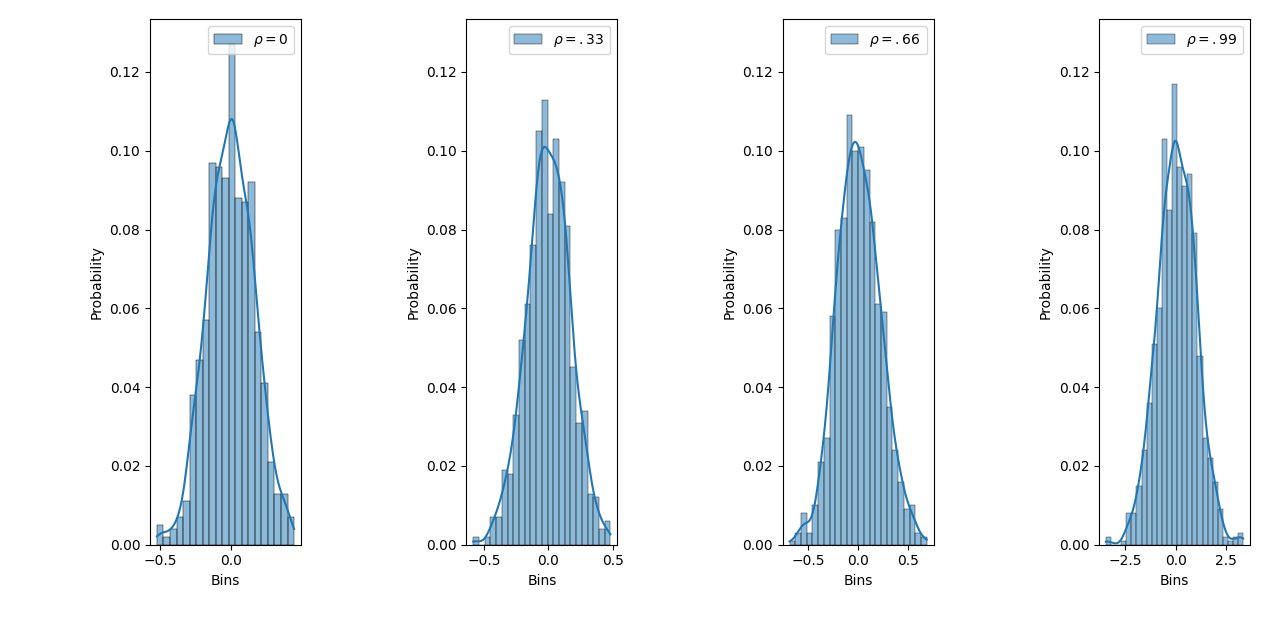}
}

\subfloat{
\includegraphics[width=15cm, height=4cm]{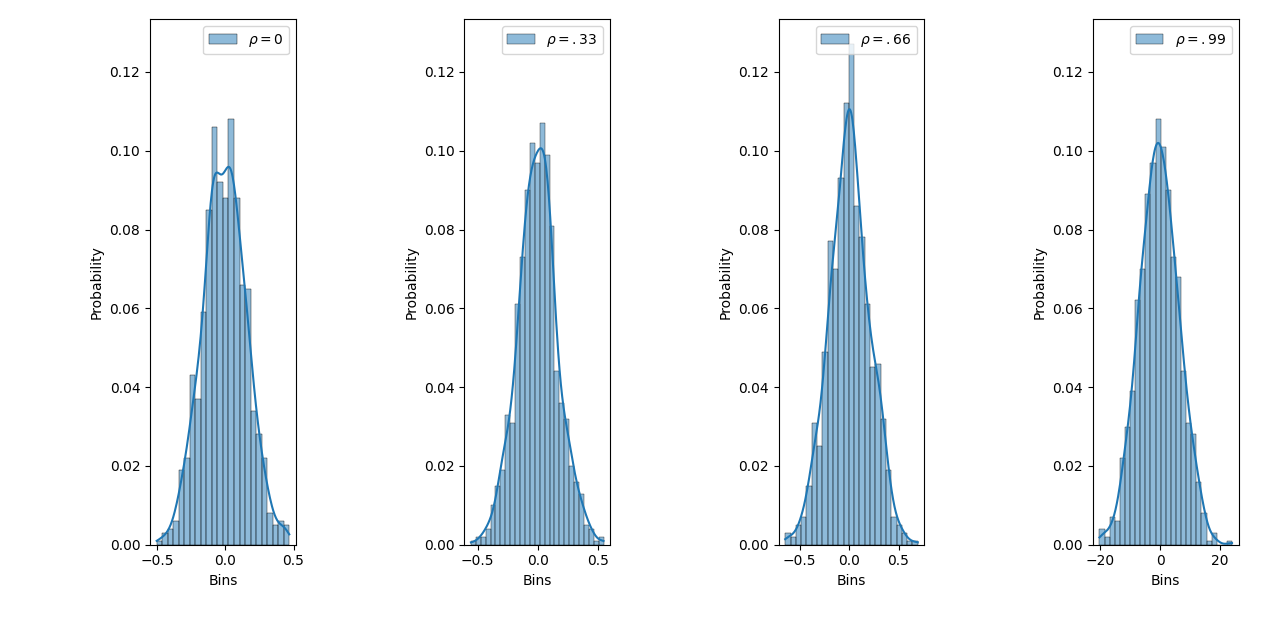}
}
	\caption{Example for the CLT for different values of $\rho$ as per the first and third coordinate projections resepectively. }
	\label{fig:fig1}	
\end{figure}
\FloatBarrier
\begin{figure}
\captionsetup[subfigure]{justification=centering}
\centering
\subfloat{
\includegraphics[width=15cm, height=4cm]{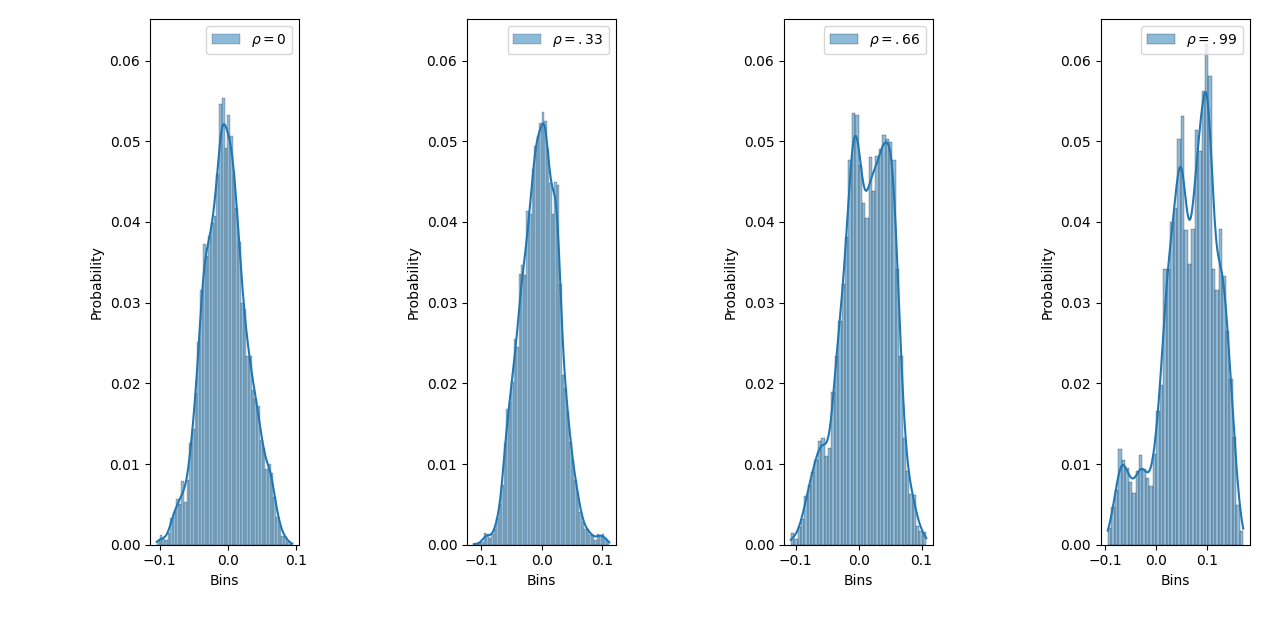}
}

\subfloat{
\includegraphics[width=15cm, height=4cm]{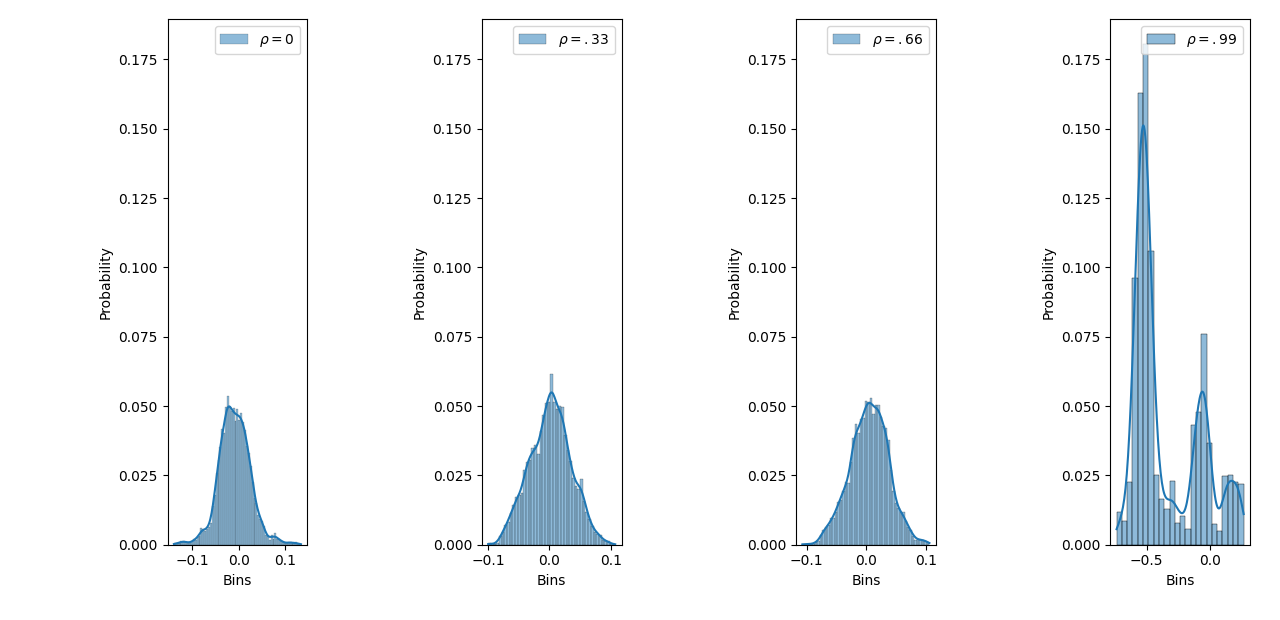}
}
	\caption{Approximate samples from the first and third coordinate marginals of the target distribution for different values of $\rho$.}
	\label{fig:fig2}	
\end{figure}
\FloatBarrier
\subsection{Simulating from Gaussian-Cosine distribution}
Consider the problem of simulating from the density \[\pi(x)\propto \exp\left(-\frac{1}{2}\left|x\right|^2+\lambda_1 \cos\left|x\right|\right)\] where $0<\lambda_1<1$. Note that this distribution has no closed form expression and hence sampling from this probability measure shall require non-trivial techniques. We shall use the preconditioned LMC algorithm to sample from this distribution. Note that, for this problem, 
\[g(x)=\frac{1}{2}\left|x\right|^2-\lambda_1 \cos\left|x\right|.\] Hence we have
\begin{align*}
\nabla g(x)&=x-\lambda_1 \sin\left|x\right|\,\frac{x}{\left|x\right|},\\
\nabla^2 g(x)&=\left(1-\lambda_1\, \frac{\sin\left|x\right|}{\left|x\right|}\right) I-\lambda_1 \frac{\cos\left|x\right|\left|x\right|-\sin\left|x\right|}{\left|x\right|^3}xx^{\mathsf{T}}.
\end{align*}
Therefore we have $M=1+\lambda_1$ which is the Lipschitz constant and $m=1-\lambda_1$ which is the coefficient of strong convexity. This is easy to see by finding the eigenvalues of the matrix $\nabla^2 g(x)$ and upper and lower bounding them. We simulate histograms of the distributions simulated for different preconditioning matrices. We consider $10^4$ iterations with $10^3$ replications for each simulation. Histograms are constructed with the average of the estimates which give the simulation results for the CLT and are plotted in Figure~\ref{fig:fig3}. In Figure~\ref{fig:fig4} we exhibit observations which are approximately sampled from $\Pi$ and draw histograms for the same. We note the dominant quadratic trend in the curve and also the fact that for $\rho$ values close to $1$, one sees that more iterations are needed to get the desired result. This is due to the fact that as $\rho$ is considered closer to $1$, the smallest eigenvalue of $\mathcal{T}$ goes to $0$ and hence more iterations are needed to detect the small gap.
\FloatBarrier
\begin{figure}
\captionsetup[subfigure]{justification=centering}
\centering
\subfloat{
\includegraphics[width=15cm, height=4cm]{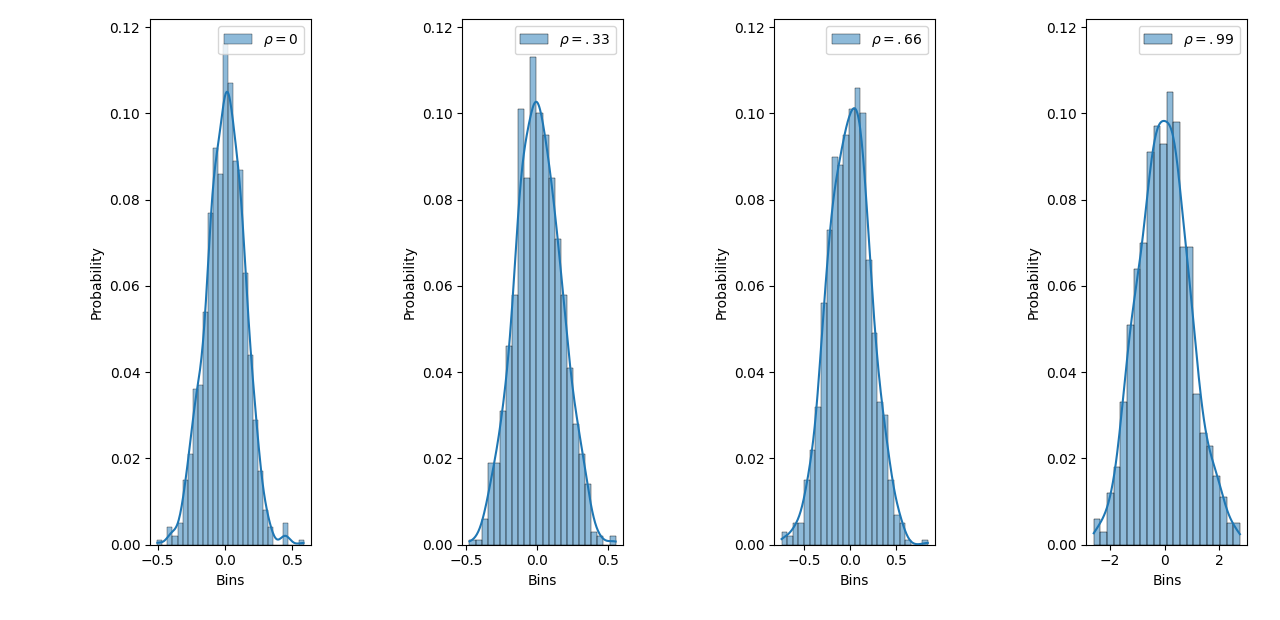}
}

\subfloat{
\includegraphics[width=15cm, height=4cm]{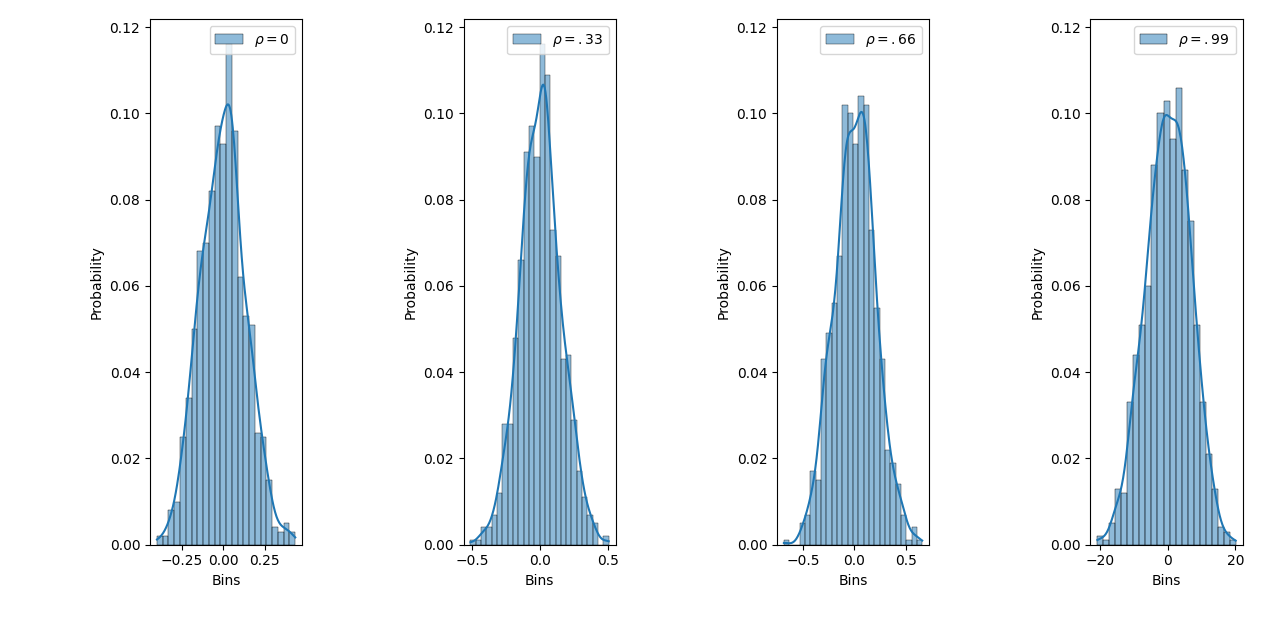}
}
	\caption{Example for the CLT for different values of $\rho$ as per the first and third coordinate projections resepectively. }
	\label{fig:fig3}	
\end{figure}
\FloatBarrier
\begin{figure}
\captionsetup[subfigure]{justification=centering}
\centering
\subfloat{
\includegraphics[width=15cm, height=4cm]{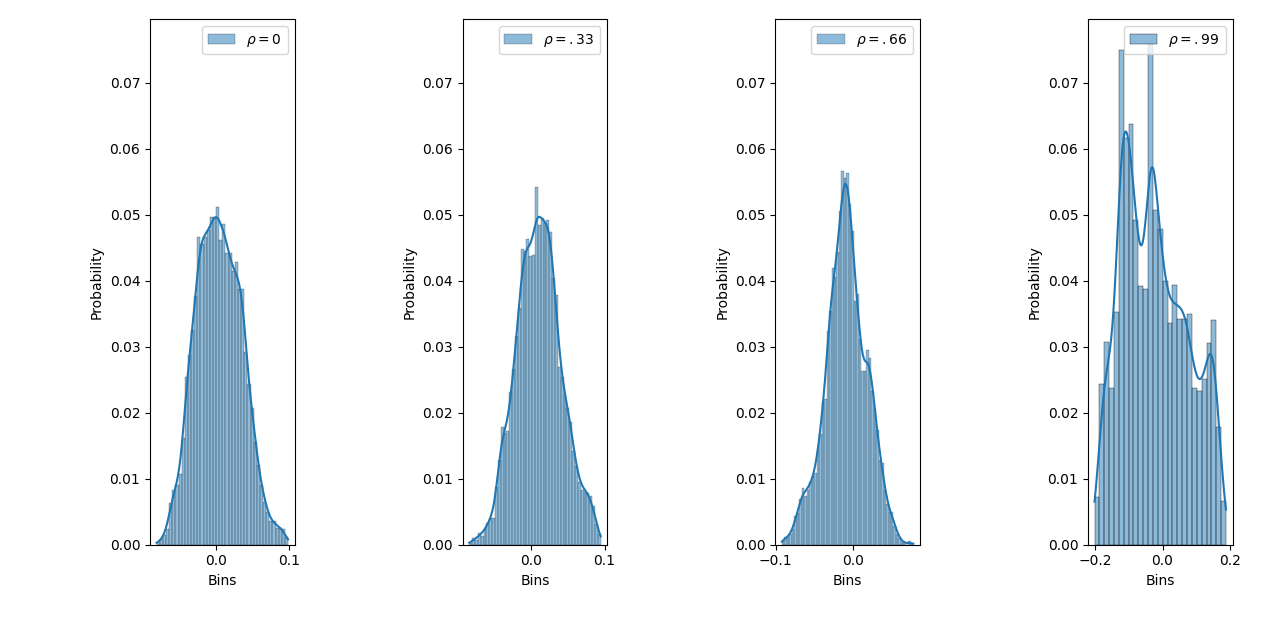}
}

\subfloat{
\includegraphics[width=15cm, height=4cm]{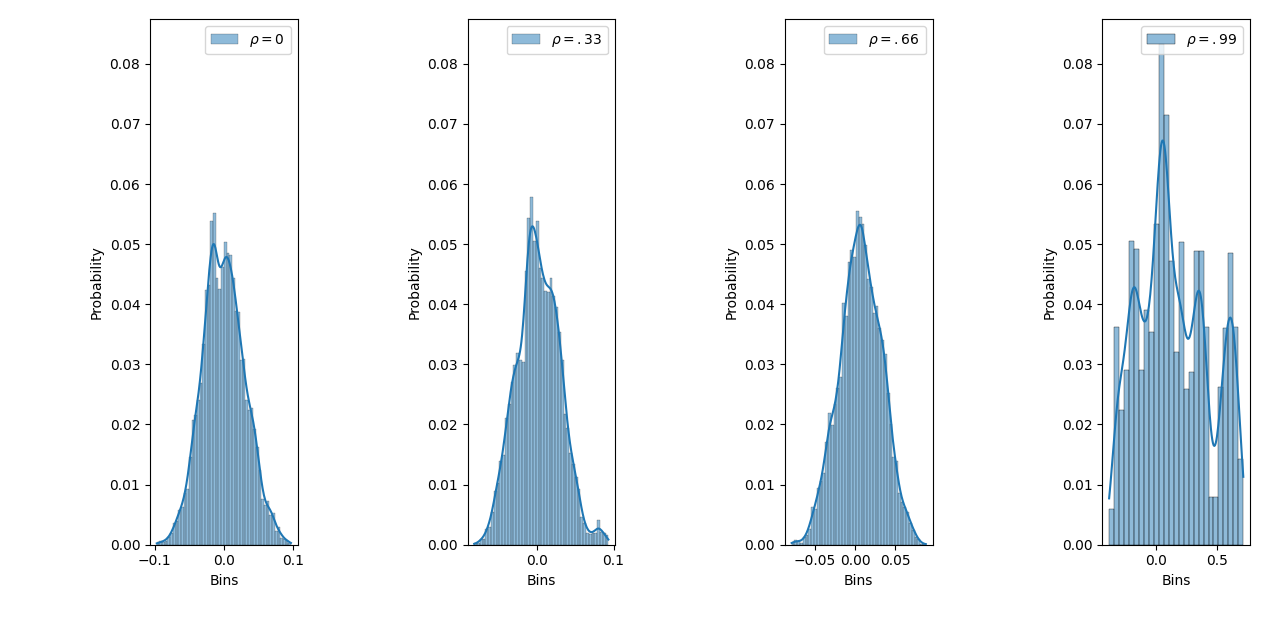}
}
	\caption{Approximate samples from the first and third coordinate marginals of the target distribution for different values of $\rho$. }
	\label{fig:fig4}	
\end{figure}
\FloatBarrier
\subsection{Simulating from a Reinforcement Learning Setup}
\begin{figure}[!ht]
    \centering
    \includegraphics[width=12cm]{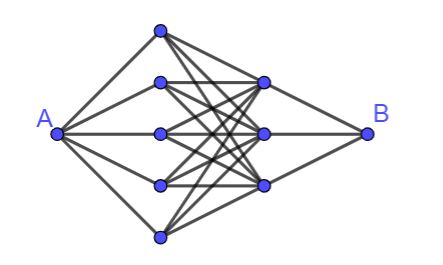}
	\caption{Traveller's Paths}
	\label{fig:figsomething}		 
\end{figure}
For this example, we consider a reinforcement learning setup where we examine the problem of a person staring from home to reach a particular destination. He traverses multiple paths while travelling to his destination and at each path he incurs a cost. This cost can be both positive and negative. The traveller wants to select a path with a reasonable cost. 

To frame the problem more mathematically consider Figure~\ref{fig:fig2}, where A is the stating point or home and B is the destination. Figure~\ref{fig:fig2} represents a graph $G$ with vertex set $V$ and edge set $E$. $\theta_{e}$ denotes the cost of traversing edge $e \in E$. We have a cutoff cost $M$, which is indicative of a desired upper bound of the total cost of travel. Our problem is selection of paths with reasonable cost. Hence we consider the model
\begin{align*}
    y_t\mid \theta \sim Ber\left(\frac{1}{1+\exp\left(\sum_{e\in x_t} \theta_e -M\right)}\right); \ t=1,2,3,\cdots, n
\end{align*}
where $t$ denotes each  instance and $n$ is the total number of  instances. $x_t$ is indicative of the path from A to B at instance $t$. $\theta$ denotes the set of all $\theta_e$ arranged as a vector. That is if we consider a number $e$ denoting each edge, then $\theta=(\theta_1,\theta_2,\cdots,\theta_{|E|})$. Here $|E|$ denotes the cardinality of the edge set. Also instances are considered to be independent of each other. We consider $\theta_e \overset{iid}{\sim} N(0,\sigma^2)$ for $e \in E$. This describes the distribution of the cost of traversing a path. We can also see that each path $x_t$ can be denoted by a vector. Let us consider with minor abuse of notation that $x_t$ is a vector and denote $x_t=(x_{t1},x_{t2},\cdots,x_{t|E|})$ where $x_{t,e}=1$ if edge $e \in E$ is included in the path. Hence we have 
\begin{align*}
    y_t \mid \theta \sim Ber\left(\frac{1}{1+\exp\left(x^{\mathsf{T}}_t \theta -M\right)}\right); \ t=1,2,\cdots,n.
\end{align*}
Hence this reduces to the Bayesian logistic regression setting. Now the question we are interested in is what is the distribution of the costs of traversing the edges given the data for all the instances. This enables us in finding the edge with the lowest average cost. We use $y$ to denote all the $y_t$ i.e., the vector of the feasibility at each instance and $X=[x_1,x_2,\cdots, x_n]_{|E|\times n}$. We have
\begin{align*}
   p( \theta \mid y,X)  \propto \left[\prod_{t=1}^{n} \left(\frac{1}{1+\exp\left(x^{\mathsf{T}}_t \theta -M\right)}\right)^{y_t} \left(\frac{\exp\left(x^{\mathsf{T}}_t \theta -M\right)}{1+\exp\left(x^{\mathsf{T}}_t \theta -M\right)}\right)^{1-y_t}\right]\ \exp\left(-\frac{1}{2\sigma^2}\left|\theta\right|^2\right).
\end{align*}
Thus we have
\begin{align*}
    \log p( \theta \mid y,X)=C+\sum_{t=1}^{n} \left(1-y_t\right)\left(x^{\mathsf{T}}_t \theta -M\right)-\sum_{t=1}^{n}\log \left(1+\exp\left(x^{\mathsf{T}}_t \theta -M\right)\right)-\frac{1}{2\sigma^2}\left|\theta\right|^2
\end{align*}
where $C$ is a constant independent of $\theta$. Hence we get
\begin{align*}
    \nabla \log p( \theta \mid y,X)=\sum_{t=1}^{n} \left(1-y_t\right)x_t-\sum_{t=1}^{n}\frac{\exp\left(x^{\mathsf{T}}_t \theta -M\right)}{1+\exp\left(x^{\mathsf{T}}_t \theta -M\right)}\,x_t-\frac{1}{\sigma^2}\,\theta
\end{align*}
and 
\begin{align*}
    \nabla^2 \log p( \theta \mid y,X)=-\sum_{t=1}^{n}\frac{\exp\left(x^{\mathsf{T}}_t \theta -M\right)}{\left(1+\exp\left(x^{\mathsf{T}}_t \theta -M\right)\right)^2}\,x_t x^{\mathsf{T}}_t-\frac{1}{\sigma^2}I.
\end{align*}
This implies that the negative log-likelihood is strongly convex. We perform simulation studies with $10^4$ iterations and $10^3$ replications. The findings are very similar to the previous examples.

\FloatBarrier
\begin{figure}
\captionsetup[subfigure]{justification=centering}
\centering
\subfloat{
\includegraphics[width=15cm, height=4cm]{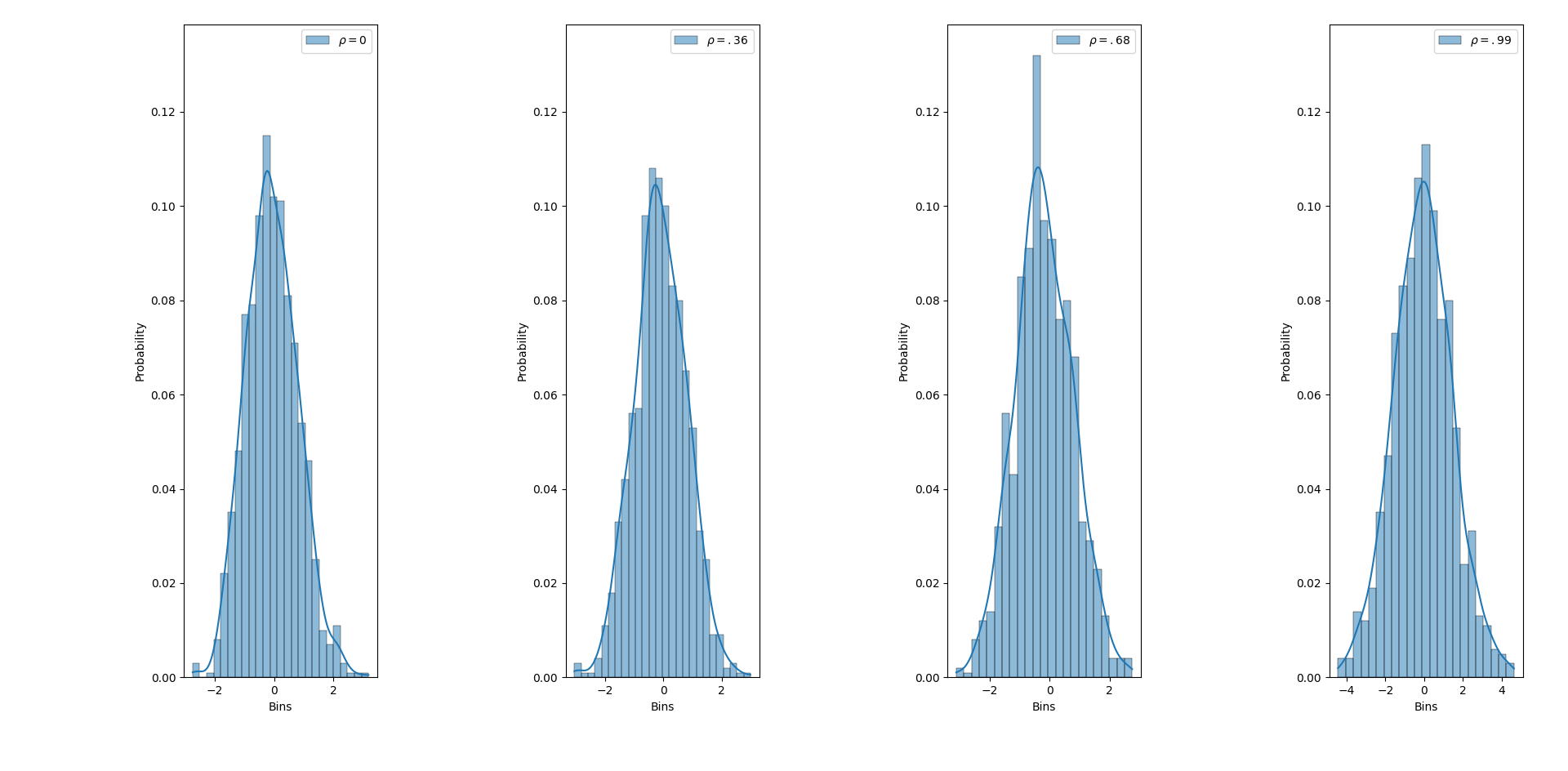}
}

\subfloat{
\includegraphics[width=15cm, height=4cm]{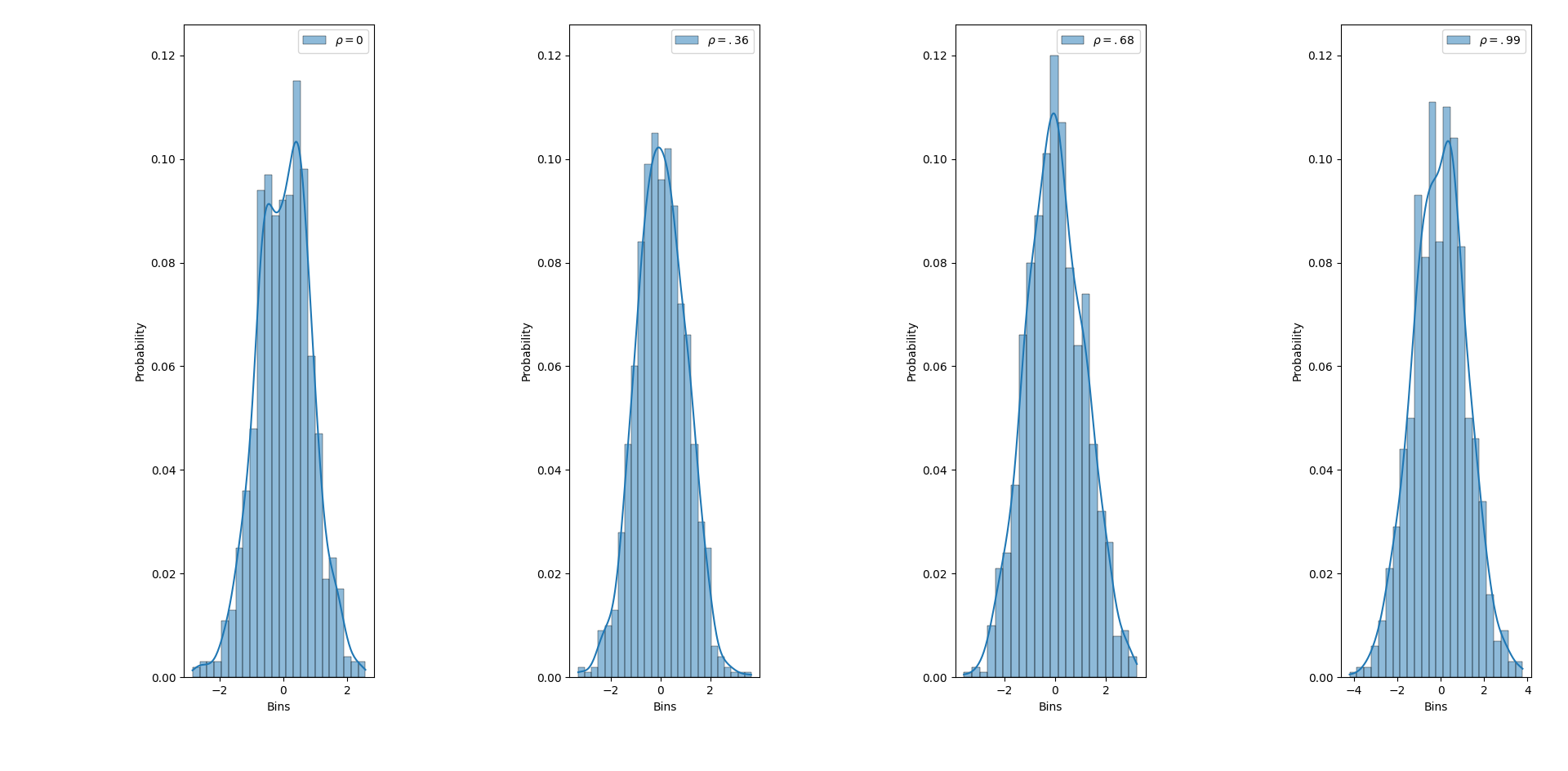}
}
	\caption{Example for the CLT for different values of $\rho$ as per the first and the fifth coordinate.}
	\label{fig:fig5}	
\end{figure}
\FloatBarrier
\begin{figure}
\captionsetup[subfigure]{justification=centering}
\centering
\subfloat{
\includegraphics[width=15cm, height=4cm]{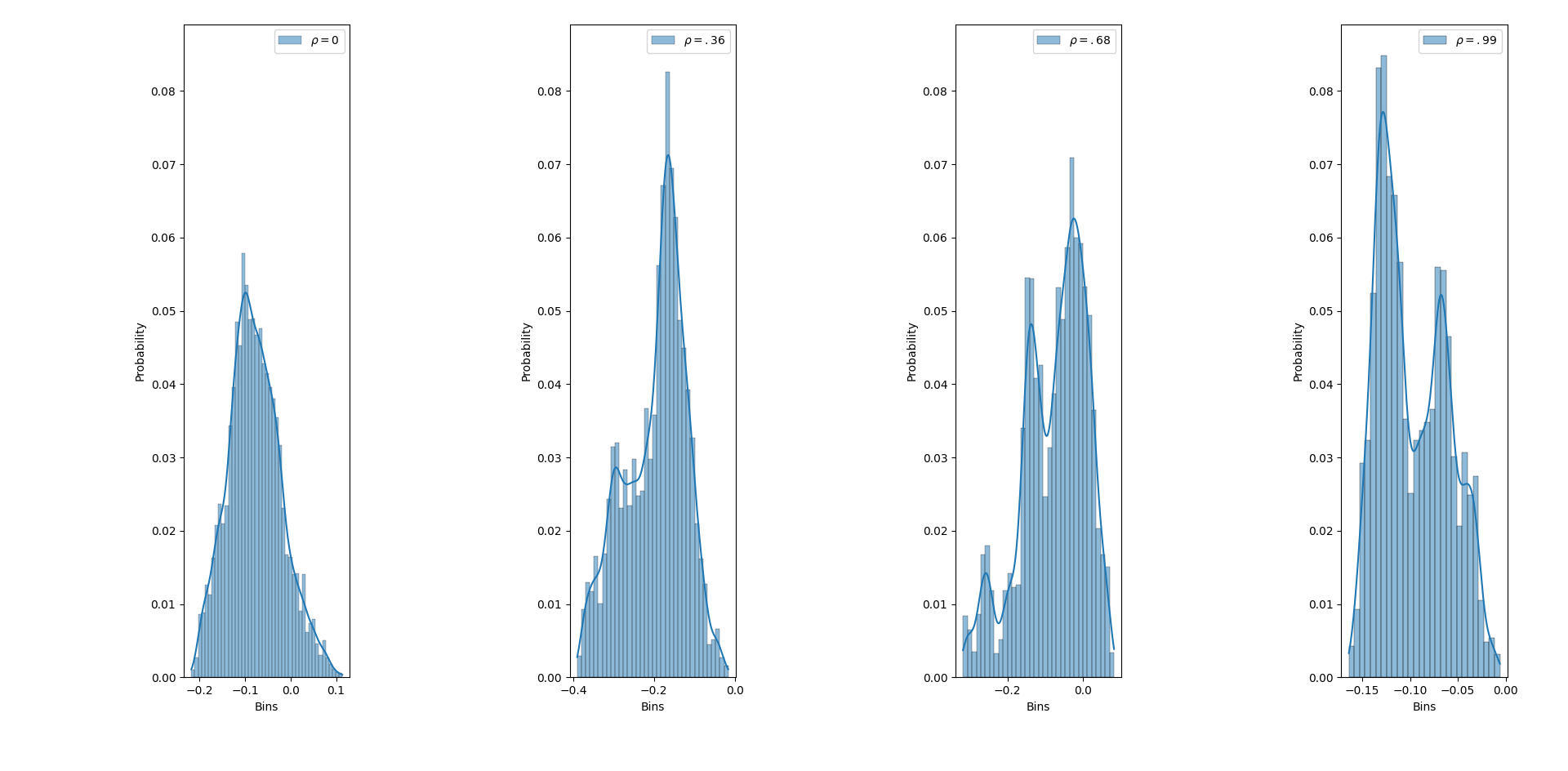}
}

\subfloat{
\includegraphics[width=15cm, height=4cm]{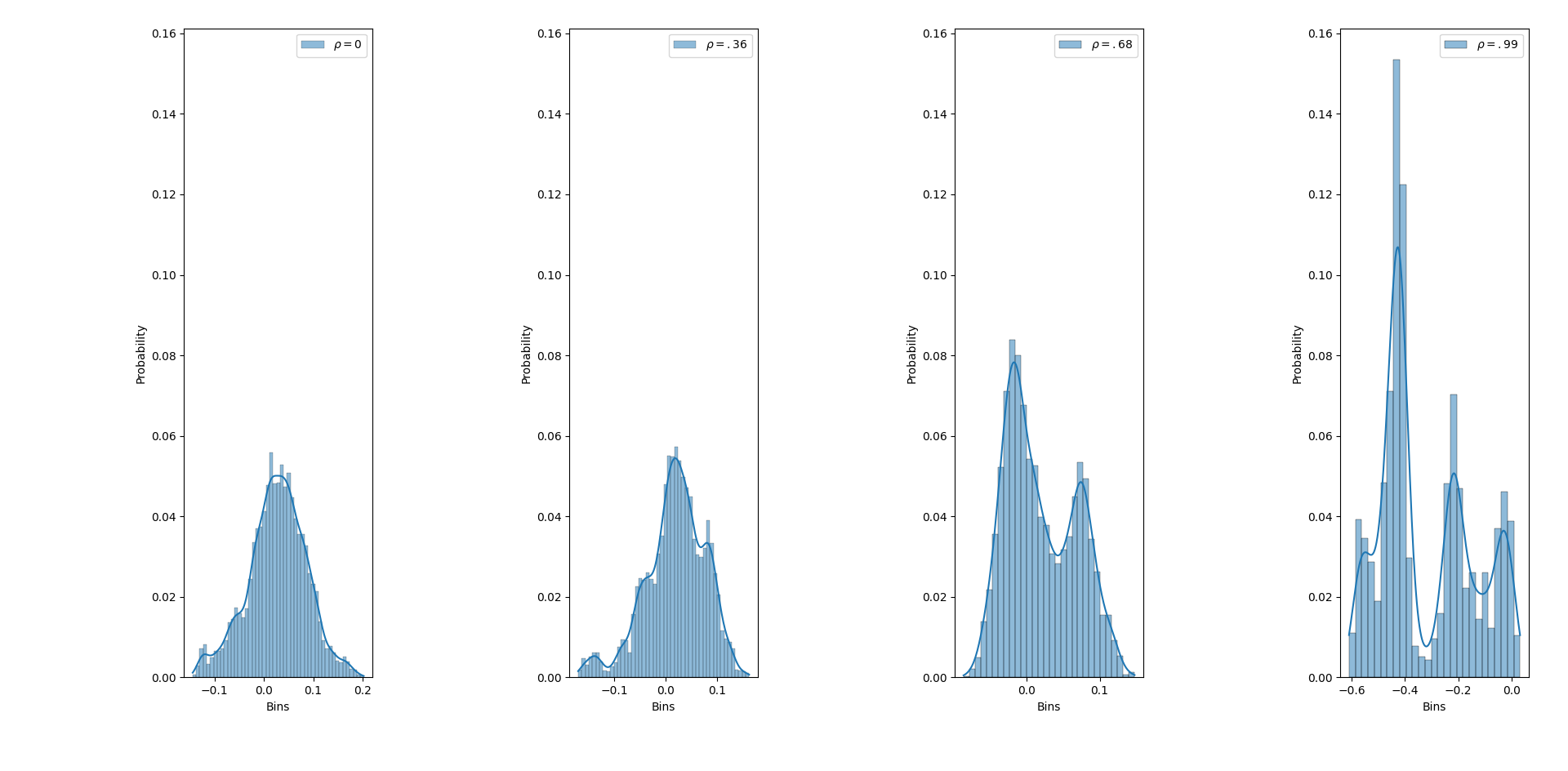}
}
	\caption{Approximate samples from the first and third coordinate marginals of the target distribution for different values of $\rho$.}
	\label{fig:fig6}	
\end{figure}
\FloatBarrier

\subsection{Concluding remarks.} 
In this paper we study the preconditioned LMC algorithm which is widely used by practitioners of fast sampling. The fast sampling bounds for the algorithm under strong conditions for KL-divergence has been settled. Also there has been numerous works in this area where there is no preconditioning. Given this, we make the following comments:
\begin{enumerate}
    \item We establish the convergence of the preconditioned LMC algorithm for general preconditioning matrices to a stationary distribution dependent on the step size in total variation. This is given in Theorem~\ref{thm:geom:erg}. 
    \item In addition to the previous point we derive \textit{explicit} convergence bounds of the algorithm to the stationary distribution dependent on the step size in total variation. This can be viewed in Proposition~\ref{prop:lwr:bnd}.
    \item We derive a CLT for preconditioned LMC samples which may be used for the purposes of statistical inference. This is exhibited in Theorem~\ref{thm:prec:clt}.
    \item In addition we also exhibit how we can use our results to conduct inference on the mode of the target distribution and how the permissible range for $\gamma$ changes when the preconditioning matrix is no longer spatially varying. This is shown in Corollary~\ref{coro:geom:erg} and Propositions~\ref{coro:geom:erg}, \ref{prop:clt:constant:H}.
    \item We establish a fast sampling bound of the preconditioned LMC algorithm to the target distribution when the preconditioning is spatially invariant, in the Wasserstein distance. This may be viewed in Theorem~\ref{Mainthm}.
    \item Simulation experiments exhibit that the fast sampling and CLT accuracy is dependent on the bounds of the preconditioning. The dependence however seems somewhat robust to minor changes and fast sampling procedures seem to exhibit more sensitivity to this change than the CLT.
    \item One interesting question is how to extend the problem of approximate sampling in the regime of strong convexity with general preconditioning matrices. This is especially challenging in a Riemannian Manifold where the concept of convexity is non-trivial.
\end{enumerate}

\bibliographystyle{plain}
\bibliography{bibfile1}

\begin{appendices}
\section{Discrete Time Analysis on the
Preconditioned Langevin Monte Carlo Algorithm }\label{sec:disc:time}

\begin{lemma}\label{drift:intermediate}
Let Assumptions~\ref{assm1} and~\ref{assm:prec:bds}	hold and let 	\[ \frac{m\, m_H}{M^2_H\, M^2}\left(1-\sqrt{1-\frac{M^2_H\, M^2\, \beta}{m^2_H\, m^2\, (1+\beta)}} \right)<\gamma< \frac{m\, m_H}{M^2_H\, M^2}\left(1+\sqrt{1-\frac{M^2_H\, M^2\, \beta}{m^2_H\, m^2\, (1+\beta)}} \right). \] Then
\begin{align*}
\mathbb{E}\left[(x_{k+1}-x^*)H^{-1}(x_k)(x_{k+1}-x^*)\mid x_k\right]\le \left(1-2\gamma \, m\,m_H +\gamma^2\, M^2_H\, M^2\right) V(x_k)+2\gamma p.
\end{align*}
\begin{proof}[Proof]
	It can be seen that in the expectation of the quadratic expression, the matrix $H(x_k)$ is fixed given the sigma-field of $x_k$. Therefore 
	\begin{align*}
	&\mathbb{E}\left[(x_{k+1}-x^*)H^{-1}(x_k)(x_{k+1}-x^*)\mid x_k\right]\\
	& \le \mathbb{E}\Bigg\{\left[(x_k-x^*)-\gamma H(x_k)\nabla g(x_k)+\sqrt{2\gamma}H^{1/2}(x_k)\xi_{k+1}\right]^{\mathsf{T}}H^{-1}(x_k)\\
	&\quad \quad \quad \quad \left[(x_k-x^*)-\gamma H(x_k)\nabla g(x_k)+\sqrt{2\gamma}H^{1/2}(x_k)\xi_{k+1}\right]\mid x_k\Bigg\}\\
	&=(x_k-x^*)^{\mathsf{T}}H^{-1}(x_k)(x_k-x^*)-2\gamma (x_k-x^*)^{\mathsf{T}}\nabla g(x_k)+\gamma^2 \nabla g(x_k)H(x_k)\nabla g(x_k)+2\gamma \mathbb{E}(\xi^{\mathsf{T}}_{k+1}\xi_{k+1})\\
	&=V(x_k)-2\gamma (x_k-x^*)^{\mathsf{T}}\nabla g(x_k)+\gamma^2 \nabla g(x_k)H(x_k)\nabla g(x_k)+2\gamma\,p.
	\end{align*}
	Observe that,
	\begin{align*}
	(x_k-x^*)^{\mathsf{T}}\nabla g(x_k)&\ge m\, |x_k-x^*|^2\\
	&\ge m\, m_H \,(x_k-x^*)H^{-1}(x_k)(x_k-x^*)\\
	&=m\,m_H\, V(x_k)
	\end{align*}
	and
	\begin{align*}
	&\nabla g(x_k)^{\mathsf{T}}H(x_k)\nabla g(x_k)\\
	&=(x_k-x^*)^{\mathsf{T}}\nabla^2 g(\xi_{x_k,x^*})H(x_k)\nabla^2 g(\xi_{x_k,x^*})(x_k-x^*)\\
	&\le M_H \|\nabla^2 g(\xi_{x_k,x^*})(x_k-x^*)\|^2\\
 & \le M_H\, M^2 \|(x_k-x^*)\|^2\\
 &\le M^2_H\, M^2 V(x_k).
	\end{align*}
	Therefore, 
\begin{align*}
	&\mathbb{E}\left[(x_{k+1}-x^*)H^{-1}(x_k)(x_{k+1}-x^*)\mid x_k\right]\\
	&\le V(x_k)-2\gamma m\,m_H\, V(x_k)+\gamma^2\,M^2 \, M^2_H\, V(x_k)+2\gamma p
\end{align*}	
and the result follows.
\end{proof}
\end{lemma}
\begin{prop}\label{prop:drift}
	Let Assumptions~\ref{assm1} and~\ref{assm:prec:bds} hold and let 
	\[ \frac{m\, m_H}{M^2_H\, M^2}\left(1-\sqrt{1-\frac{M^2_H\, M^2\, \beta}{m^2_H\, m^2\, (1+\beta)}} \right)<\gamma< \frac{m\, m_H}{M^2_H\, M^2}\left(1+\sqrt{1-\frac{M^2_H\, M^2\, \beta}{m^2_H\, m^2\, (1+\beta)}} \right). \] Then, for the Lyapunov function $V(x)=(x-x^*)^{\mathsf{T}}H^{-1}(x)(x-x^*)$ and algorithm \eqref{generalized:lmc},
	\begin{align*}
	\mathbb{E}\left[V(x_{k+1})\mid x_k\right]\le \tilde{\lambda}\, V(x_k)+b
	\end{align*}
	for some $0<\tilde{\lambda}<1$, $b>0$ and $\gamma>0$.
\end{prop}
\begin{proof}[Proof]
	Our proof shall be two-fold. We shall first reduce the given problem into a simpler one and then solve the simpler problem. Note that 
	\begin{align*}
	&\mathbb{E}\left[V(x_{k+1})\mid x_k\right]\\
	&=\mathbb{E}\left[(x_{k+1}-x^*)^{\mathsf{T}}H^{-1}(x_{k+1})(x_{k+1}-x^*)-(x_{k+1}-x^*)^{\mathsf{T}}H^{-1}(x_k)(x_{k+1}-x^*)\mid x_k\right]\\
	&\quad \quad +\mathbb{E}\left[(x_{k+1}-x^*)^{\mathsf{T}}H^{-1}(x_k)(x_{k+1}-x^*)\mid x_k\right].
	\end{align*}
	Now, for the first term
	\begin{align*}
	&\mathbb{E}\left[(x_{k+1}-x^*)^{\mathsf{T}}H^{-1}(x_{k+1})(x_{k+1}-x^*)-(x_{k+1}-x^*)^{\mathsf{T}}H^{-1}(x_k)(x_{k+1}-x^*)\mid x_k\right]\\
	&\le \mathbb{E}\left[\left|(x_{k+1}-x^*)^{\mathsf{T}}H^{-1}(x_{k+1})(x_{k+1}-x^*)-(x_{k+1}-x^*)^{\mathsf{T}}H^{-1}(x_k)(x_{k+1}-x^*)\right|\mid x_k\right]\\
	&\le \mathbb{E}\left[\left|x_{k+1}-x^*\right|^2\, \left\|H^{-1}(x_{k+1})-H^{-1}(x_k)\right\|_2\mid x_k\right]\\
	&\le  \frac{\beta}{M_H}\,\mathbb{E}\left[\left|x_{k+1}-x^*\right|^2\mid x_k\right] \\
	&\le \beta \,\mathbb{E}\left[\left(x_{k+1}-x^*\right)^{\mathsf{T}}H^{-1}(x_k)\left(x_{k+1}-x^*\right)\mid x_k\right].
	\end{align*}
By using Lemma~\ref{drift:intermediate}, we get
	\begin{align*}
	\mathbb{E}(V(x_{k+1})\mid x_k) &\le\left(1+\beta\right) \left[\left(1-2\gamma \, m\,m_H +\gamma^2\, M^2_H\, M^2\right) V(x_k) +2\gamma p\right].
	\end{align*}
	We finish the proof.
\end{proof}
\begin{coro}\label{coro:drift}
	Let the conditions of Proposition~\ref{prop:drift} hold. For the Lyapunov function $\tilde{V}(x)=V(x)+1$, there exists the following drift condition
	\begin{align*}
	P\tilde{V}(x)\le \tilde{\lambda}\, \tilde{V}(x)+\tilde{b}
	\end{align*}
	with $\tilde{b}=b+1-\tilde{\lambda}$.
\end{coro}
\begin{proof}
	Note that
	\begin{align*}
	P(V+1)(x)&=PV(x)+1\\
	&\le \tilde{\lambda}\, V(x)+b+1\\
	&\le \tilde{\lambda}\, (V(x)+1)+(b+1-\tilde{\lambda}).
	\end{align*}
 The proof is completed. 
\end{proof}

\begin{lemma}\label{lemma:Harris}
	Let Assumptions~\ref{assm1} and \ref{assm:prec:bds} hold. Then, the Markov chain as defined by \eqref{precalgorithm} is aperiodic, irreducible with respect to the Lebesgue measure and Harris recurrent as 
	\[ \frac{m\, m_H}{M^2_H\, M^2}\left(1-\sqrt{1-\frac{M^2_H\, M^2\, \beta}{m^2_H\, m^2\, (1+\beta)}} \right)<\gamma< \frac{m\, m_H}{M^2_H\, M^2}\left(1+\sqrt{1-\frac{M^2_H\, M^2\, \beta}{m^2_H\, m^2\, (1+\beta)}} \right). \]
\end{lemma}
\begin{proof}
	We consider three cases separately. We first tackle the question of irreducibility. Note that to establish irreducibility, we need to exhibit that 
	\begin{align*}
	P(x_{k+1}\in A \mid x_k)>0
	\end{align*}
 for any set $A$ with $\lambda^{Leb}(A)>0$. 
	This is trivially true as 
	\[x_{k+1}\mid x_k \sim N\left(\gamma H(x_k)\nabla g(x_k),2\gamma H(x_k) \right) \]
	where $H(\theta)$ is of full rank for any $\theta$. Thus $\{x_k\}_{k\ge 0}$ is indeed irreducible with respect to the Lebesgue measure. Next we show that the chain is aperiodic. This can also be seen very easily as if the chain is not aperiodic, there exists a partition of $\mathbb{R}^p$ as $\{D_i\}_{i=1}^{n}$, such that $\cup_{i=1}^{n} D_i=\mathbb{R}^p$ and $P(x_{k+1} \in D^c_i\mid x_k \in D_i)=1$. This is impossible as \[x_{k+1}\mid x_k \sim N\left(\gamma H(x_k)\nabla g(x_k),2\gamma H(x_k) \right) .\]  Lastly, by Corollary 4 from \cite{mattingly2002ergodicity}, we know  a drift condition implies that the resultant chain is Harris. We then have established the Harris recurrence.
\end{proof}
\begin{coro}\label{coro:exitence:stat}
	For any 	\[ \frac{m\, m_H}{M^2_H\, M^2}\left(1-\sqrt{1-\frac{M^2_H\, M^2\, \beta}{m^2_H\, m^2\, (1+\beta)}} \right)<\gamma< \frac{m\, m_H}{M^2_H\, M^2}\left(1+\sqrt{1-\frac{M^2_H\, M^2\, \beta}{m^2_H\, m^2\, (1+\beta)}} \right)\] there exists a stationary measure $\pi_{\gamma}$ for the Markov chain as defined in \eqref{precalgorithm}.
\end{coro}
\begin{proof}
	The proof is an immediate consequence of Lemma~\ref{lemma:Harris}.
\end{proof}
\begin{lemma}\label{lemma:minorization}
	Let the conditions of Proposition~\ref{prop:drift} hold. Consider the set \[C=\left\{x:\, \tilde{V}(x)\le\frac{2\, \tilde{b}}{\alpha
		-\tilde{\lambda}} \right\}\] for any $\alpha \in (\lambda,1)$. For $x \in C$, we have 
	\begin{align*}
	P(x,A)\ge \eta \,\nu(A)
	\end{align*}
	for some $\eta>0$ and probability measure $\nu$.
\end{lemma}
\begin{proof}
Define $\nu(\cdot)$ as the uniform measure restricted to $C$, i.e., for any $A \in \mathcal{B}(\mathbb{R}^p)$, we have 
\begin{align*}
\nu(A)=\frac{\mu^{Leb}(A\cap C)}{\mu^{Leb}(C)}.
\end{align*}	
For $x \in C$ and $A \subset C$, we have
\begin{align*}
P(x,A)&=\int_A \frac{1}{(2\pi)^{p/2}\det(H(x))^{1/2}}\, \exp\left\{-\frac{1}{2}\left(y-H(x)\nabla g(x)\right)^{\mathsf{T}}H^{-1}(x)\left(y-H(x)\nabla g(x)\right)\right\}\, dy\\
& \ge \frac{1}{(2\pi)^{p/2}M_H^{p/2}} \, \inf_{y \in C, x\in C}\exp\left\{-\frac{1}{2}\left(y-H(x)\nabla g(x)\right)^{\mathsf{T}}H^{-1}(x)\left(y-H(x)\nabla g(x)\right)\right\}\int_A dy\\
&\ge \frac{\mu^{Leb}(C)}{(2\pi)^{p/2}M_H^{p/2}} \, \inf_{y \in C, x\in C}\exp\left\{-\frac{1}{2}\left(y-H(x)\nabla g(x)\right)^{\mathsf{T}}H^{-1}(x)\left(y-H(x)\nabla g(x)\right)\right\} \frac{\mu^{Leb}(A)}{\mu^{Leb}(C)}\\
&\ge \eta \, \nu(A).
\end{align*}
Hence we complete the proof. 
\end{proof}

 With the above preparation we are now ready to present the proofs of Theorem~\ref{thm:geom:erg}, Propositions~\ref{coro:geom:erg} and \ref{prop:lwr:bnd}.

\begin{proof}[Proof of Theorem~\ref{thm:geom:erg}]
	The proof is a direct consequence of Proposition~\ref{prop:drift} with Lemma~\ref{lemma:minorization} along with Theorem 12 from \cite{rosenthal1995minorization}.
\end{proof}
\begin{proof}[Proof of Proposition~\ref{coro:geom:erg}]
    By Lemma~\ref{drift:intermediate}, we know that 
    \begin{align*}
        \mathbb{E}\left[V(x_{k+1})\mid x_k\right] \le \left(1-2\gamma m\,m_H+\gamma^2\, M^2_H\, M^2\right)\, V(x_k)+2\gamma \, p
    \end{align*}
    for any $0<\gamma< \frac{2\, m_H \, m}{M^2_H\, M^2}$. Therefore, we have 
    \[P\tilde{V}(x)\le \tilde{\lambda}\, \tilde{V}(x)+\tilde{b} \]
    where $\tilde{\lambda}=\left(1-2\gamma m\,m_H+\gamma^2\, M^2_H\, M^2\right)$, $b=2\gamma p$,  $\tilde{b}=2\gamma p+1-\tilde{\lambda}$ and $\tilde{V}(x)=V(x)+1$. Note that this result holds with $0<\gamma< \frac{2\, m_H \, m}{M^2_H\, M^2}$. This implies that the drift condition holds with $\gamma$ in the given interval. This also implies that Lemma~\ref{lemma:Harris} and Corollary~\ref{coro:exitence:stat} hold with $0<\gamma< \frac{2\, m_H \, m}{M^2_H\, M^2}$. Hence, the proof follows by the same argument as Theorem~\ref{thm:geom:erg}.
\end{proof}
\begin{proof}[Proof of Proposition~\ref{prop:lwr:bnd}]
    Note that if we can give a lower bound for \[\inf_{y \in C, x\in C}\exp\left\{-\frac{1}{2}\left(y-H(x)\nabla g(x)\right)^{\mathsf{T}}H^{-1}(x)\left(y-H(x)\nabla g(x)\right)\right\}\] then we  indeed finish the proof  by Lemma~\ref{lemma:minorization}. In fact, by noting that 
    \[M^{-1}_H\|x-x^*\|^2\le\tilde{V}(x)\le \frac{2\tilde{b}}{\alpha-\tilde{\lambda}}.\]
    Thus, $C$ is contained in the closed ball with center at $x^*$ and radius $\sqrt{M_H(2\tilde{b}/(\alpha-\tilde{\lambda})-1)}$. We refer to this set as $B(x^*)$. Consequently,
    \begin{align*}
        & \inf_{y \in C, x\in C}\exp\left\{-\frac{1}{2}\left(y-H(x)\nabla g(x)\right)^{\mathsf{T}}H^{-1}(x)\left(y-H(x)\nabla g(x)\right)\right\}\\
        & \ge \inf_{y \in B(x^*), x\in B(x^*)}\exp\left\{-\frac{1}{2}\left(y-H(x)\nabla g(x)\right)^{\mathsf{T}}H^{-1}(x)\left(y-H(x)\nabla g(x)\right)\right\}.
    \end{align*}
     Now, on $B(x^*)$, we have
     \begin{align*}
         &\left(y-H(x)\nabla g(x)\right)^{\mathsf{T}}H^{-1}(x)\left(y-H(x)\nabla g(x)\right)\\
         &= y^{\mathsf{T}} H^{-1}(x) y-2 \, y^{\mathsf{T}}\nabla g(x) +\nabla g(x)^{\mathsf{T}} H(x) \nabla g(x)\\
         &\le \frac{2}{m_H}\left(\|y-x^*\|^2+\|x^*\|^2\right)+2\,\|y\|\|\nabla g(x)\|+M_H\, M^2\|x-x^*\|^2\\
         &\le \frac{2\, M_H}{m_H}\left(\frac{2\tilde{b}}{\alpha -\tilde{\lambda}}-1\right) + \frac{2\|x^*\|^2}{m_H}+2\, M \sqrt{M_H\left(\frac{2\tilde{b}}{\alpha -\tilde{\lambda}}-1\right)}\|y\|+M^2_H\,M^2 \left(\frac{2\tilde{b}}{\alpha -\tilde{\lambda}}-1\right)\\
         &\le \frac{2\, M_H}{m_H}\left(\frac{2\tilde{b}}{\alpha -\tilde{\lambda}}-1\right) + \frac{2\|x^*\|^2}{m_H}+2\, M\, M_H \left(\frac{2\tilde{b}}{\alpha -\tilde{\lambda}}-1\right)\\
         &\quad \quad \quad +2\, M\, \|x^*\| \sqrt{M_H\left(\frac{2\tilde{b}}{\alpha -\tilde{\lambda}}-1\right)}+M^2_H\,M^2 \left(\frac{2\tilde{b}}{\alpha -\tilde{\lambda}}-1\right)\\
         &=\left(\frac{2\tilde{b}}{\alpha -\tilde{\lambda}}-1\right)\left(\frac{2\,M_H}{m_H}+2\, M\,M_H +M^2_H \, M^2\right)+\frac{2\|x^*\|^2}{m_H}+2\, M\, \|x^*\| \sqrt{M_H\left(\frac{2\tilde{b}}{\alpha -\tilde{\lambda}}-1\right)}.
     \end{align*}
     Therefore, 
     \begin{align*}
         &\inf_{y \in C, x\in C}\exp\left\{-\frac{1}{2}\left(y-H(x)\nabla g(x)\right)^{\mathsf{T}}H^{-1}(x)\left(y-H(x)\nabla g(x)\right)\right\} \\
         &\ge \exp\left\{-\left(\frac{2\tilde{b}}{\alpha -\tilde{\lambda}}-1\right)\left(\frac{\,M_H}{m_H}+\, M\,M_H +\frac{M^2_H\, M^2}{2} \right)-\frac{\|x^*\|^2}{m_H}- M\, \|x^*\| \sqrt{M_H\left(\frac{2\tilde{b}}{\alpha -\tilde{\lambda}}-1\right)}\right\}
     \end{align*}
     and the proof is completed. 
\end{proof}
\section{Continuous Time Analysis of Preconditioned Langevin Monte Carlo Algorithm}
\subsection{Continuous Time Analysis of Preconditioned Langevin} 
	Note that we call \eqref{precdiffusion} a diffusion as by \cite[Theorem 5.2.1]{oksendal2013stochastic}. 
	Define as $P_t$ the transition semi-group of the Markov chain associated with \eqref{precdiffusion}. Also use $\mathcal{A}$ to define the generator of $P_t$. We also know that the domain of the generator is any $f \in C^2(\mathbb{R}^p)$.
	\begin{lemma}\label{lemmafirst}
		Under Assumptions~\ref{assm1} and~\ref{assm:prec:bds} we have 
		\begin{align*}
		P_t V(x)  \le \frac{1}{m_H}\left[\exp{\left(-2m\,m_Ht\right)}\left|x-x^*\right|^2+\frac{p}{m} \left(1-\exp{\left(-2m\,m_Ht\right)}\right)\right]
		\end{align*}
		where $V(x)=(x-x^*)^{\mathsf{T}}H^{-1}(x-x^*)$.
	\end{lemma}
	\begin{proof}
		We know that
		\begin{align}\label{generatoreqn}
		\mathcal{A}f(y)=-\left\langle H\nabla g(y), \nabla f(y)\right \rangle+\left \langle \nabla ,H\nabla \right \rangle \left(f(y)\right).
		\end{align}
		For any positive measurable function $V$ define $v(t,x)=P_t V(x)$.
		Using Dynkin's formula (see Oskendal~\cite{oksendal2013stochastic}), we know 
		\begin{align*}
		\frac{\partial v(t,x)}{\partial t}= P_t \mathcal{A}V(x).
		\end{align*}
		Define \[V(x)=(x-x^*)^\mathsf{T}H^{-1}(x-x^*).\]	
		In this case 
		\begin{align*}
		\mathcal{A}V(x)&=-\left \langle H\nabla g(x)-H\nabla g(x^*),2\, H^{-1}(x-x^*)\right\rangle +2\, p\\
		&=-2\,\left(\nabla g(x)-\nabla g(x^*)\right)^{\mathsf{T}}(x-x^*)+2\,p.
		\end{align*}
		By strong convexity, we have \(\left(\nabla g(x)-\nabla g(x^*)\right)^{\mathsf{T}}(x-x^*)\ge m\, \left|x-x^*\right|^2\). This implies 
		\begin{align*}
		\mathcal{A}V(x)& \le -2m\, \left|x-x^*\right|^2+2\,p.
		\end{align*}	
		Now 
		\begin{align*}
		\frac{\partial v(t,x)}{\partial t}& = P_t \mathcal{A}V(x)\\
		& \le -2m \,P_t \left|x-x^*\right|^2 +2\,p\\
		& \le -2m \,m_H\,P_t V(x)+2\,p.
		\end{align*}
		The last line here again follows as $|x-x^*|^2 \ge m_H\,(x-x^*)^{\mathsf{T}}H^{-1}(x-x^*)$.
		This implies 
		\begin{align*}
		\frac{\partial v(t,x)}{\partial t} & \le -2m \, m_H \, v(t,x)+ 2\,p.
		\end{align*}
		Using the Gronwall Lemma, 
		\begin{align*}
		v(t,x) \le \exp{\left(-2m\,m_Ht\right)}\left(x-x^*\right)^{\mathsf{T}}H^{-1}\left(x-x^*\right)+\frac{p}{m\,m_H} \left(1-\exp{\left(-2m\,m_Ht\right)}\right).
		\end{align*}
		Therefore, by the fact \(\left(x-x^*\right)^{\mathsf{T}}H^{-1}\left(x-x^*\right) \le \frac{1}{m_H}|x-x^*|^2\), we obtain
		\[v(t,x) \le \frac{1}{m_H}\left[\exp{\left(-2m\,m_Ht\right)}\left|x-x^*\right|^2+\frac{p}{m} \left(1-\exp{\left(-2m\,m_Ht\right)}\right)\right]\]
		and the desired result follows.	
	\end{proof}
As mentioned previously, $P_t$ is the Markov transition kernel and using Theorem~\ref{thmfoxma} we know that the Preconditioned Langevin diffusion~\eqref{precdiffusion} has a stationary distribution, which we refer to as $\Pi$. Note that for the stationary measure have $\Pi P_t=\Pi$ by definition.
	\begin{lemma}\label{lemmastationary}
		Under Assumptions~\ref{assm1} and~\ref{assm:prec:bds} we have
		\[\Pi V\le \frac{p}{\kappa} \]
		where $\kappa=m\,m_H$.
	\end{lemma}
	\begin{proof}
		We know $\Pi P_t=\Pi$. This implies for a fixed constant $c>0$, we have 
		\begin{align*}
		\Pi\left(V\wedge c\right)&= \Pi P_t\left(V\wedge c\right)\\
		&\le \Pi\left(P_tV \wedge c\right)\\
		&\le \Pi\left\{\frac{1}{m_H}\left[e^{-2\kappa t}\left|x-x^*\right|^2+\frac{p}{m}\left(1-e^{-2\kappa t}\right)\right]\wedge c\right\}.
		\end{align*}
		The second line follows as $\wedge$ is a concave function and the next line follows from Lemma~\ref{lemmafirst}. Using DCT as $t \to \infty$, one has 
		\[\Pi\left(V\wedge c\right)\le \frac{p}{\kappa}\]
		for any $c>0$. Taking $c\to \infty$, by the monotone convergence theorem, we obtain the result.
	\end{proof}
	Next we present a lemma which exhibits a contraction for the $t$ step markov kernel when the starting points are different. Recall the Wasserstein distance between two measures $\mu$ and $\nu$ as defined in \eqref{Tie}.
	\begin{lemma}\label{wassersteincontraction}
		Under Assumptions~\ref{assm1} and~\ref{assm:prec:bds} we have
		\begin{align*}
		W_2(\delta_{X_0}P_t,\delta_{Y_0}P_t)\le \exp{\left(-\kappa t \right)}\, \frac{1}{\sqrt\kappa_*}\left|X_0-Y_0\right|
		\end{align*}
		where $\kappa_*:=m_H/M_H$ is the condition number of $H$.
	\end{lemma}
	\begin{proof}[Proof of Lemma~\ref{wassersteincontraction}]
		Consider the stochastic differential equations
		\begin{align*}
		dX_t&=-H\nabla g(X_t)dt+\sqrt{2}H^{1/2}dB_t\\
		dY_t&=-H\nabla g(Y_t)dt+\sqrt{2}H^{1/2}dB_t.
		\end{align*}
		Here the Brownian motions are the same. The starting points of $X_t$ and $Y_t$  are $X_0$ and $Y_0$, respectively. 
		Evidently, 
		\begin{align*}
		dH^{-1}X_t&=-\nabla g(X_t)dt+\sqrt{2}H^{-1/2}dB_t\\
		dH^{-1}Y_t&=-\nabla g(Y_t)dt+\sqrt{2}H^{-1/2}dB_t.
		\end{align*}	
		This implies 
		\begin{align*}
		\left\langle dH^{-1}\left(X_t-Y_t\right),X_t-Y_t\right\rangle=-\left\langle \left(\nabla g(X_t)-\nabla g(Y_t)\right),X_t-Y_t\right\rangle dt
		\end{align*}
		which in turn implies 
		\begin{align*}
		\frac{1}{2}\left(X_t-Y_t\right)^\mathsf{T}H^{-1}\left(X_t-Y_t\right)&=\frac{1}{2}\left(X_0-Y_0\right)^\mathsf{T}H^{-1}\left(X_0-Y_0\right)-\int_{0}^{t}\left\langle \nabla g(X_s)-\nabla g(Y_s),X_s-Y_s\right\rangle ds\\
		&\le \frac{1}{2}\left(X_0-Y_0\right)^\mathsf{T}H^{-1}\left(X_0-Y_0\right)-m\,\int_{0}^{t}\left|X_s-Y_s\right|^2 ds.
		\end{align*}
		Now, \[\frac{1}{m_H}\,\left|X_t-Y_t\right|^2\ge \frac{1}{2}\left(X_t-Y_t\right)^\mathsf{T}H^{-1}\left(X_t-Y_t\right) \ge \frac{1}{M_H}\,\left|X_t-Y_t\right|^2.\]
		Consequently, 
		\begin{align*}
		\left|X_t-Y_t\right|^2 \le \frac{M_H}{m_H}\left|X_0-Y_0\right|^2-2m\,M_H\,\int_{0}^{t}\left|X_s-Y_s\right|^2 ds.
		\end{align*}
		By the Gronwall lemma, we have
		\begin{align*}
		\left|X_t-Y_t\right|^2 &\le \exp{\left(-2t\, m \, M_H\right)}\, \frac{M_H}{m_H}\left|X_0-Y_0\right|^2.
		\end{align*}
		Taking expectation and using the definition of the Wasserstein distance we have 
		\begin{align*}
		W_2(\delta_{X_0}P_t,\delta_{Y_0}P_t)\le \exp{\left(-t\, m \, M_H\right)}\, \left(\frac{M_H}{m_H}\right)^{1/2}\left|X_0-Y_0\right|.
		\end{align*}
		The proof is concluded.
	\end{proof}
This lemma exhibits a contraction in the measures at the t-th time step starting at different point masses. Next we present the proof of Proposition~\ref{conttimeprop} which is one of the key results presented in this work.
	\begin{proof}[Proof of Proposition~\ref{conttimeprop}]
		We start with the triangle inequality of the Wasserstein distance	
		\begin{align*}
		W_2\left(\delta_x P_t,\Pi\right)\le W_2\left(\delta_x P_t,\delta_{x^*} P_t\right)+W_2\left(\delta_{x^*} P_t,\Pi\right).
		\end{align*}
		Using Lemma~\ref{wassersteincontraction}, we have
		\begin{align*}
		W_2\left(\delta_x P_t,\delta_{x^*} P_t\right) \le \sqrt{\frac{M_H}{m_H}} e^{-\kappa t}\left|x-x^*\right|.
		\end{align*}
		Hence, we have 
		\begin{align*}
		W_2\left(\delta_x P_t,\Pi\right) & \le \sqrt{\frac{M_H}{m_H}} e^{-\kappa t}\left|x-x^*\right|+W_2\left(\delta_{x^*} P_t,\Pi\right)\\
		&\le \sqrt{\frac{M_H}{m_H}} e^{-\kappa t}\left|x-x^*\right|+\left(\mathbb{E}_{\Pi}\left(W^2_2(\delta_{x^*} P_t,\delta_x P_t)\right)\right)^{1/2}.
		\end{align*}
		Now, again from Lemma~\ref{lemmastationary} and the fact that $V(x)\ge |x-x^*|^2$, we have
		\begin{align*}
		\frac{M_H\, p}{m_H\, m}\ge M_H\, \Pi V \ge \Pi \left|x-x^*\right|^2.
		\end{align*}
		Hence
		\begin{align*}
		W_2\left(\delta_x P_t,\Pi\right) \le \sqrt{\frac{M_H}{m_H}} e^{-\kappa t}\left|x-x^*\right| + \sqrt{\frac{M_H}{m_H}} e^{-\kappa t}\left(\frac{p\, M_H}{m\, m_H}\right)^{1/2}.
		\end{align*}
	\end{proof}
\subsection{Discrete Time Approximation}
	\begin{lemma}\label{gradientlemma}
		Let  $g(x)$ be a  Lipschitz function defined on $\mathbb{R}^p$ with  Lipschitz  constant $M$. Then we have
		\begin{align*}	\sum_{k=0}^{[T/\gamma]}\mathbb{E}\left|\nabla g(x_k)\right|^2 \le \frac{g(x_0)-g(x^*)+M\, d \, T M_H}{\left(m_H\gamma-\frac{M\gamma^2}{2}M^2_H\right)}
		\end{align*} 
		where $\gamma<2\,m_H/ (M M^2_H)$.
	\end{lemma}
	\begin{proof}
		First, 
		\begin{align*}
		g(x_{k+1})&=g\left(x_k-\gamma H\nabla g(x_k)+\sqrt{2\gamma} H^{1/2}\xi_{k+1}\right)\\
		&\le g(x_k)-\left\langle \nabla g(x_k),\gamma H\nabla g(x_k)-\sqrt{2\gamma}H^{1/2}\xi_{k+1}\right\rangle+\frac{M}{2}\left|-\gamma H\nabla g(x_k)+\sqrt{2\gamma}H^{1/2}\xi_{k+1}\right|^2.
		\end{align*}
		By taking expectation, the previous step implies  
		\begin{align*}
		\mathbb{E}\left(g(x_{k+1})\right)\le \mathbb{E}\left(g(x_{k})\right)- \left(m_H\gamma-\frac{M\gamma^2}{2}M^2_H\right)\mathbb{E}\left|\nabla g(x_k)\right|^2+M\, p\gamma M_H.
		\end{align*}
		Hence, 
		\begin{align*}
		\left(m_H\gamma-\frac{M\gamma^2}{2}M^2_H\right)\mathbb{E}\left|\nabla g(x_k)\right|^2 \le \mathbb{E}\left(g(x_{k})\right)-\mathbb{E}\left(g(x_{k+1})\right)+M\, p\gamma M_H.
		\end{align*}
		Therefore we have
		\begin{align*}
		\sum_{k=0}^{[T/\gamma]}\mathbb{E}\left|\nabla g(x_k)\right|^2 \le \frac{g(x_0)-g(x^*)}{\left(m_H\gamma-\frac{M\gamma^2}{2}M^2_H\right)}+\frac{M\, p \, T M_H}{\left(m_H\gamma-\frac{M\gamma^2}{2}M^2_H\right)}.
		\end{align*}
		Hence the proof is completed.
	\end{proof}
	\begin{prop}\label{KLprop}
	Suppose  Assumptions~\ref{assm1} and~\ref{assm:prec:bds} hold. Recall  \eqref{precdiffusion} and \eqref{discrprocess}. We then  have
		\begin{align*}
		KL\left(\mathbb{P}_{X_t}\mid \mid \mathbb{P}_{D_t}\right)& \le M^3_H\,\frac{M\,\gamma^2}{12}\frac{\left(g(x_0)-g(x^*)+M\, p \, T M_H\right)}{\left(m_H-\frac{M\gamma}{2}M^2_H\right)}+M^2_H\,\frac{M}{4} p\, T\gamma
		\end{align*} 
		where $\gamma<2\,m_H/(M M^2_H)$.
	\end{prop}
	\begin{proof}
		Note that \eqref{precdiffusion} and \eqref{discrprocess} have the same Brownian Motion and also admit strong solutions. Hence we can apply Girsanov's Theorem \cite{aries2013statistics} to obtain 
		\begin{align}\label{KLdiveq}
		KL\left(\mathbb{P}_{X_t}\mid \mid \mathbb{P}_{D_t}\right)&=\frac{1}{4}\sum_{k=0}^{[T/\gamma]}\int_{k\gamma}^{(k+1)\gamma} \mathbb{E}\left|H^{1/2}\nabla g(D_t)-H^{1/2}\nabla g(D_k)\right|^2.
		\end{align}
	 We follow in the footsteps of \cite[Lemma 2]{dalalyan2017theoretical}. 
		Given the identity \eqref{KLdiveq}, we have 
		\begin{align*}
		KL\left(\mathbb{P}_{D_t}\mid \mid \mathbb{P}_{X_t}\right)&=\frac{1}{4}\sum_{k=0}^{[T/\gamma]}\int_{k\gamma}^{(k+1)\gamma} \mathbb{E}\left|H^{1/2}\nabla g(D_t)-H^{1/2}\nabla g(D_k)\right|^2\\
		&\le \frac{1}{4}\sum_{k=0}^{[T/\gamma]}\int_{k\gamma}^{(k+1)\gamma} \left|\left|H \right|\right|_2\,\mathbb{E}\left|\nabla g(D_t)-\nabla g(D_k)\right|^2\\
		&\le M_H\,\frac{M}{4}\sum_{k=0}^{[T/\gamma]} \int_{k\gamma}^{(k+1)\gamma} \mathbb{E}\left|D_t -D_k\right|^2\\
		&\le M_H \,\frac{M}{4}\sum_{k=0}^{[T/\gamma]} \int_{k\gamma}^{(k+1)\gamma} \mathbb{E}\left|-H\nabla g(D_k)\left(t-k\gamma\right)+\sqrt{2}H^{1/2}\left(B_t-B_{k\gamma}\right)\right|^2.
		\end{align*}
		Also 
		\begin{align*}
		&\mathbb{E}\left|-H\nabla g(D_k)\left(t-k\gamma\right)+\sqrt{2}H^{1/2}\left(B_t-B_{k\gamma}\right)\right|^2\\
		&=\mathbb{E}\left|-H\nabla g(D_k)\left(t-k\gamma\right)\right|^2+\mathbb{E}\left|\sqrt{2}H^{1/2}\left(B_t-B_{k\gamma}\right)\right|^2\\
		&=\left(t-k\gamma\right)^2\mathbb{E}\left|H\nabla g(D_k)\right|^2+2\,\mathbb{E}\left(B_t-B_{k\gamma}\right)^{\mathsf{T}}H\left(B_t-B_{k\gamma}\right)\\
		& \le M^2_H\left(t-k\gamma\right)^2\mathbb{E}\left|\nabla g(D_k)\right|^2+2\, \left(t-k\gamma\right) trace(H)\\
		&\le M^2_H\left(t-k\gamma\right)^2\mathbb{E}\left|\nabla g(D_k)\right|^2+2\, p\, \left(t-k\gamma\right) M_H.
		\end{align*}
	    Therefore,
		\begin{align*}
		KL\left(\mathbb{P}_{D_t}\mid \mid \mathbb{P}_{X_t}\right)& \le
		M^3_H\,\frac{M\,\gamma^3}{12}\sum_{k=0}^{[T/\gamma]}\mathbb{E}\left|\nabla g(D_k)\right|^2+M_H\,\frac{M}{4}\sum_{k=0}^{[T/\gamma]}p\, \gamma^2\, M_H.
		\end{align*}
		Using Lemma~\ref{gradientlemma}, we have,
		\begin{align*}
		KL\left(\mathbb{P}_{D_t}\mid \mid \mathbb{P}_{X_t}\right)& \le M^3_H\,\frac{M\,\gamma^2}{12}\frac{\left(g(x_0)-g(x^*)+M\, p \, T M_H\right)}{\left(m_H-\frac{M\gamma}{2}M^2_H\right)}+M^2_H\,\frac{M}{4} p\, T\gamma.
		\end{align*}
  The proof is finished. 
	\end{proof}
	
 \begin{remark}\label{shaa}
Note that Proposition~\ref{KLprop} implies that for $\gamma < \kappa_*/(M M_H)$. Therefore, \[KL\left(\mathbb{P}_{D_t}\mid \mid \mathbb{P}_{X_t}\right)=C^*\gamma\]
		where $C^*$ is a constant independent of $\gamma$. 
	Note that we can take $C^*$ as
	\begin{align*}
C^*=\frac{M_H}{6}\, \left(g(x_0)-g(x^*)\right)+\frac{5}{12}\,M^2_H\, MpT.
	\end{align*}
 \end{remark}
 
	We shall connect the KL-divergence with the Wasserstein metric. 
Let $X$ be a Polish space and  $\mathcal{P}(X)$ be the space of all Borel probability measures on $X$. 
Recall $W_{\tilde{p}}(\mu,\nu)$ the Wasserstein-$\tilde{p}$ metric between the probability measures $\mu$ and $\nu$
defined by
	\[W_{\tilde{p}}(\mu,\nu)=\inf_{\tilde{\gamma}\in \Gamma(\mu,\nu)}\left(\int d(x,y)^{\tilde{p}} d\tilde{\gamma}(x,y)\right)^{1/\tilde{p}}\]
	where $\Gamma(\mu,\nu)$ denotes the set of coupling  between $\mu$ and $\nu$ defined in \eqref{Tie}. For the case when $\tilde{p}=2$, this is the canonical Wasserstein-Monge-Kantorvich distance stated in \eqref{Tie}. 
	\begin{lemma}\label{wassKLlemma}\cite[Corollary 3]{bolley2005weighted}
		Let $X$ be a space equipped with a metric distance $d$. Let $\tilde p\ge 1$ and let $\nu$ be a Borel probability measure on $X$. Assume that there exists an $x_0 \in X$ and $\alpha>0$ such that \(\int e^{\alpha d(x_0,x)^{2\tilde{p}}} d\nu(x)\) is finite. Then \[\quad W_{\tilde{p}}(\mu,\nu) \le C\left[KL(\mu\mid \nu)^{1/\tilde{p}}+\left(\frac{KL(\mu\mid \nu)}{2}\right)^{1/2\tilde{p}}\right]\]
	for any $\mu \in \mathcal{P}(X)$, 	where 
		\[C=2 \inf_{x_0 \in X, \alpha>0} \left[\frac{1}{\alpha} \left(\frac{3}{2}+\log \int e^{\alpha d(x_0,x)^{\tilde{p}}} d\nu(x) \right)\right].\]
	\end{lemma}
	We shall use Lemma~\ref{wassKLlemma} with $\mu=\mathbb{P}_{D_t}$, $\nu=\mathbb{P}_{X_t}$, $d(x,y)=|x-y|$, $\tilde p=2$ and $X=\mathbb{R}^p$. Note that to use Lemma~\ref{wassKLlemma}, we must establish the moment condition which is equivalent to establishing \(\mathbb{E}\left(e^{\alpha |X_t|^2}\right)<\infty\) for some $\alpha>0$. Note that we can take any $\alpha$ and $ x_0=0$ as $C$ is defined as the minimum value.
Define $L_t=\exp(\alpha X^{\mathsf{T}}_t H^{-1}X_t)$.
\begin{lemma}\label{MomentLemma}
	Under Assumptions~\ref{assm1} and~\ref{assm:prec:bds} with \(0<\alpha<\kappa/2 \), we have
\begin{align*}
\mathbb{E}\left(L_t\right) \le e^{2\alpha p T }\left[E\left(L_0\right)+\frac{\alpha \left|\nabla g(0)\right|^2}{\left(2 m- \frac{4\alpha}{m_H}\right)}\, T\right]
\end{align*}
for any $t$ and $T$ with $0\leq t \leq T$.
\end{lemma}
\begin{proof}[Proof]
	Recall that \[dX_t=-H\nabla g(X_t) dt+\sqrt{2}H^{1/2}dB_t. \] Hence by the rules of Ito differentials $dX_t dX^{\mathsf{T}}_t=2 H dt.$ 
	Again, by the rules of Ito calculus, we obtain
	\begin{align*}
	dL_t=2\, \alpha L_t \left\langle H^{-1}X_t, dX_t \right\rangle +2\,\alpha^2 L_t dX^{\mathsf{T}}_t H^{-1}X_tX^{\mathsf{T}}H^{-1}dX_t+\alpha L_t dX^{\mathsf{T}}_t H^{-1}dX_t.
	\end{align*}
	By using the definition of $X_t$ and by using the trace function on differentials, we get
	\begin{align*}
	dX^{\mathsf{T}}_t H^{-1}X_tX^{\mathsf{T}}H^{-1}dX_t&=X^{\mathsf{T}}_t H^{-1}dX_tdX^{\mathsf{T}}_t H^{-1}X_t\\
	&=2\, X^{\mathsf{T}}_t H^{-1}X_t dt.
	\end{align*}
	Using the definition of $X_t$ again to see
	\begin{align*}
	dL_t=-2\alpha \left \langle X_t, \nabla g(X_t)\right\rangle dt +2\sqrt{2}\alpha L_t X^{\mathsf{T}}_t H^{-1/2}dB_t +4\alpha^2 L_t X^{\mathsf{T}}H^{-1}X_t dt +2\alpha \, p L_t dt.
	\end{align*}
	Note that for any $x\in \mathbb{R}^p$, we have 
	\begin{align*}
	\left\langle x, \nabla g(x)\right\rangle&=\left\langle x-0, \nabla g(x)-\nabla g(0)\right\rangle +\left\langle x, \nabla g(0)\right\rangle\\
	&\ge m\left|x\right|^2-\left|x\right|\left|\nabla g(0)\right|.
	\end{align*}
Here the last line follows from the strong convexity and Cauchy-Schwartz inequality. Hence 
\begin{align*}
4\alpha^2 x^{\mathsf{T}}H^{-1}x-2\alpha x^{\mathsf{T}}\nabla g(x) &\le \frac{4\alpha^2}{m_H}   \left|x\right|^2 - 2\alpha m \left|x\right|^2+2\alpha \left|\nabla g(0)\right|\left|x\right|\\
&\le -\left(2\alpha m -\frac{4\alpha^2}{m_H}\right)\left(\left|x\right|-\frac{\alpha \left|\nabla g(0)\right|}{\left(2\alpha m- \frac{4\alpha^2}{m_H}\right)}\right)^2\\
&\quad \quad +\frac{\alpha^2 \left|\nabla g(0)\right|^2}{\left(2\alpha m- \frac{4\alpha^2}{m_H}\right)}\\
&\le \frac{\alpha \left|\nabla g(0)\right|^2}{\left(2 m- \frac{4\alpha}{m_H}\right)}.
\end{align*}
Thus
\begin{align*}
L_t \le L_0 +\frac{\alpha \left|\nabla g(0)\right|^2}{\left(2 m- \frac{4\alpha}{m_H}\right)}\, t+2\sqrt{2}\alpha \int_{0}^{t} L_s X^{\mathsf{T}}_sH^{-1/2} dB_s +2\alpha p \int_{0}^{t}L_s ds.
\end{align*}
Taking expectation we find that
\begin{align*}
\mathbb{E}\left(L_t\right) \le \mathbb{E}\left(L_0\right) +\frac{\alpha \left|\nabla g(0)\right|^2}{\left(2 m- \frac{4\alpha}{m_H}\right)}\, t+2\sqrt{2}\alpha \mathbb{E}\left( \int_{0}^{t} L_s X^{\mathsf{T}}_sH^{-1/2} dB_s\right) +2\alpha p \int_{0}^{t}\mathbb{E}\left(L_s\right) ds.
\end{align*}
Noting that $M_t=\int_{0}^{t} L_s X^{\mathsf{T}}_sH^{-1/2} dB_s$ is a continuous time square-integrable martingale with $\mathbb{E}(M_t)=0$.
As a consequence,
\begin{align*}
\mathbb{E}\left(L_t\right) \le \mathbb{E}\left(L_0\right) +\frac{\alpha \left|\nabla g(0)\right|^2}{\left(2 m- \frac{4\alpha}{m_H}\right)}\, t +2\alpha p \int_{0}^{t}\mathbb{E}\left(L_s\right) ds.
\end{align*}
By using the Gronwall Lemma, we obtain
\begin{align*}
\mathbb{E}\left(L_t\right) &\le \exp\left\{2\alpha t p\right\}\left(E\left(L_0\right)+\frac{\alpha \left|\nabla g(0)\right|^2}{\left(2 m- \frac{4\alpha}{m_H}\right)}\, t\right)\\
& \le \exp\left\{2\alpha T p\right\}\left(E\left(L_0\right)+\frac{\alpha \left|\nabla g(0)\right|^2}{\left(2 m- \frac{4\alpha}{m_H}\right)}\, T\right).
\end{align*}
The proof is completed. 
\end{proof}
\begin{coro}
	Under the setting of Lemma~\ref{MomentLemma}, we have \[\exp\left\{\frac{\alpha}{M_H}\left|X_t\right|^2\right\}\le \exp\left\{2\alpha T p\right\}\left(E\left(L_0\right)+\frac{\alpha \left|\nabla g(0)\right|^2}{\left(2 m- \frac{4\alpha}{m_H}\right)}\, T\right).\]
\end{coro}
\begin{proof}
	The proof is immediate from the inequality \[\frac{1}{M_H}\left|x\right|^2 \le x^{\mathsf{T}}H^{-1}x.\]
\end{proof}

\subsection{Proof of Theorem~\ref{Mainthm}}

\begin{proof}[Proof of Theorem~\ref{Mainthm}]
    By the definition of $T$, we have 
    \[W_2(\mathbb{P}_{X_t}, \Pi) \le \frac{\epsilon}{2}\] due to Proposition~\ref{conttimeprop}. Also by the definitions of $C,C^*,\gamma$, we have 
    \[W_2(\mathbb{P}_{X_t},\mathbb{P}_{D_t}) \le \frac{\epsilon}{2}\]
    by using Proposition~\ref{KLprop} and Lemma~\ref{wassKLlemma} in order. The result follows by the triangle inequality of the Wasserstein distance.
\end{proof}
\end{appendices}

\end{document}